\documentclass[onecolumn,11pt]{IEEEtran}

\usepackage{amsmath}
\usepackage{amssymb}
\usepackage{algorithm,algorithmicx}
\usepackage{algpseudocode}
\usepackage{caption}
\usepackage{multirow}
\usepackage{colortbl}
\usepackage{cleveref}
\usepackage{float}
\usepackage{subcaption}

\usepackage{mathptmx}
\usepackage{graphicx}
\usepackage{xspace}
\usepackage{paralist} 
\usepackage{xypic}
\xyoption{curve}

\usepackage{latexsym}
\usepackage{theorem}
\usepackage{ifthen}
\usepackage{pstricks}
\input{Qcircuit.tex}
\psset{unit=1pt,fillcolor=white,arrowscale=1.2,arrowinset=0}
\usepackage{booktabs}
\usepackage{microtype}

{\theoremheaderfont{\it}
\theorembodyfont{\rmfamily}
\newtheorem{theorem}{Theorem}
\newtheorem{definition}[theorem]{Definition}
}

\newcommand{\REVS}{{\textsc{Revs}}}
\newcommand{\LIQUID}{{LIQ{\em ui}$|\rangle$}}
\newcommand{\mypar}[1]{{\textbf{#1.}}}
\newcommand{\nix}[1]{{}}

\usepackage{listings}
\usepackage{upquote}
 \usepackage{color}

\definecolor{greencomments}{rgb}{0,0.5,0}

\lstdefinelanguage{FSharp}%
{morekeywords={let, new, match, with, rec, open, module, namespace, type, of, member, %
and, for, while, true, false, in, do, begin, fun, function, return, yield, try, %
mutable, if, then, else, cloud, async, static, use, abstract, interface, inherit, finally },
  otherkeywords={ let!, return!, do!, yield!, use!, var, from, select, where, order, by },
  keywordstyle=\color{blue},
  sensitive=true,
  basicstyle=\ttfamily,
	breaklines=true,
  xleftmargin=\parindent,
  aboveskip=\bigskipamount,
	tabsize=4,
  morecomment=[l][\color{greencomments}]{///},
  morecomment=[l][\color{greencomments}]{//},
  morecomment=[s][\color{greencomments}]{{(*}{*)}},
  morestring=[b]",
  showstringspaces=false,
  literate={`}{\`}1,
  stringstyle=\color{red},
}

\title{Reversible circuit compilation with space constraints}
\author{Alex Parent, Martin Roetteler, and Krysta M.~Svore\thanks{This
  work was carried out while the first author was a summer intern with the Quantum Architectures and Computation Group (QuArC), Microsoft Research, Redmond, WA 98052, U.S.A. Corresponding author email: martinro@microsoft.com}
}

\begin{document}

\maketitle

\begin{abstract}
\noindent 
We develop a framework for resource efficient compilation of higher-level programs into lower-level reversible circuits.
Our main focus is on optimizing the memory footprint of the resulting reversible networks.
This is motivated by 
the limited availability of qubits for the foreseeable future.
We apply three main techniques to keep the number of required qubits small when computing classical, irreversible computations by means of reversible networks: first, wherever possible we allow the compiler to make use of in-place functions to modify some of the variables.
Second, an intermediate representation is introduced that allows to trace data dependencies within the program, allowing to clean up qubits early. This 
realizes an analog to ``garbage collection'' for reversible circuits.
Third, we use the concept of so-called pebble games to transform irreversible programs into reversible programs under space constraints, allowing for data to be erased and recomputed if needed.

We introduce \REVS, a compiler for reversible circuits that can translate a subset of the functional programming language F$\#$ into Toffoli networks which can then be further interpreted for instance in \LIQUID, a domain-specific language for quantum computing and which is also embedded into F$\#$.
We discuss a number of test cases that illustrate the advantages of our approach 
including reversible implementations of SHA-2 and other cryptographic hash-functions, reversible integer arithmetic, as well as a test-bench of combinational circuits used in classical circuit synthesis.
Compared to Bennett's method, \REVS{} can reduce space complexity by a factor of $4$ or more, while having an only moderate increase in circuit size as well as in the time it takes to compile the reversible networks.
\end{abstract}



%
%

\section{Introduction}\label{intro}

The observation that arbitrary computations can be carried out by a computational device in such a way that in principle each time-step can be reversed---allowing to recover the input from the output of a computation that has been orchestrated in such a fashion--- goes back to Bennett's landmark paper \cite{Bennett:73}.
While the original motivation for reversible computing was to demonstrate that the amount of heat generated by irreversible gates---as implied by Landauer's principle---can in principle be avoided by making computations that never erase any information, it transpires that compared to the actual energy dissipation of modern integrated chips, this saving in energy is quite small, see, e.g., the survey \cite{Markov:2014}. For modern chips, the amount of energy savings due to avoiding erasure of information would be more than $10$ orders of magnitude smaller than the amount of energy savings that arise from other dissipative processes that heat up the chip. Aside from an adiabatic regime where chips would operate at ultra-low power, yet ultra-slow, therefore arguably the main application of reversible computing is therefore in quantum computing, namely as a vehicle that allows a quantum computer to carry out any function that a classical computer might carry out.

It should be noted that the ability to compute classical functions is at the core of many interesting quantum algorithms, including Shor's algorithm for discrete log and factoring \cite{Shor:97} where the reversible computations are arithmetic operations in suitable algebraic data structures such as rings and fields. Another example is Grover's algorithm \cite{Grover:96}, where reversible computations feature as the operations required to implement the predicate that implicitly defines the solution of the search problem at hand.
Many variations of this general theme exist including quantum walk algorithms that allow to traverse graphs faster than classical algorithms can, in some cases even exponentially faster, as well as some algorithms for simulation of Hamiltonians, where reversible computations may be needed for the efficient accessing of the matrix elements of the underlying Hamiltonian \cite{BCC+:2014}. 

While this may illustrate the need for techniques to turn classical computations into quantum circuits, it also may serve as an illustration of the difficulties that such a translation will present to an compiler system that aims at supporting this translation from a classical computation. Such a classical description could be, say, a program expressed in a higher-level programming language such as C or Haskell. Among the difficulties are the following issues: (i) qubits that are used as intermediate scratch space during the computation have to be cleaned up at the end of the computation or else the interference effects, on which quantum computations typically rely heavily, would disappear which would render the computation useless, (ii) the number of qubits that are needed for scratch space grows linearly with the number of classical instructions if a simple method for turning the irreversible computation into a reversible one is used such as the original Bennett method \cite{Bennett:73}.
What is more, the simple methods for making circuits reversible are extremely inefficient regarding the load-factor of the computation, namely they lead to circuit that only manipulate a tiny subset of the qubits at a given time and leave the bulk of the qubits idle.
This is particularly troublesome as early quantum computers will contain relatively few qubits. In turn, they will neither be able to afford ample scratch space nor scratch space that stays idle for large portions of the computation. 
Happily, Bennett already pointed out a way out of this dilemma by studying time-space trade-offs for reversible computation, and by introducing reversible pebble games which allow to systematically study ways to save on scratch space at the expense of recomputing intermediate results.
To determine the best pebbling strategy for the dependency graph imposed by actual real-world programs, however, is a non-trivial matter.
In this paper we follow a pragmatic approach: we are interested in solutions that work in practice and allow to handle programs at scale.
As a guiding example we consider cryptographic hash-functions such as SHA-2, which can be thought of a Boolean function $f:\{0,1\}^N \rightarrow \{0,1\}^n$, where $n\ll N$ that has a simple and straightforward classical program for its evaluation that has no branchings and only uses simple Boolean functions such as XOR, AND, and bit rotations, however, which has internal {\em state} between rounds.
The fact that there is state prevents the Boolean function to be decomposed, thereby making purely truth-table or BDD based synthesis methods useless for this problem.
On the other hand, scalable approaches such as combinator-based rewriting of the input program as a classical circuit, followed by applying the Bennett method, also run into issues for SHA-2 as the number of rounds is high and because of the large number of scratch qubits per each round, the overall required space by such methods is high.

\mypar{Related work} There are several programming languages for quantum circuits, including \LIQUID~\cite{WS:2014}, Quipper~\cite{GLR+:2013a,GLR+:2013b}, and various other approaches \cite{Omer:2000,AG:2009,SV:2009,LST+:2013,LR:2013,HPJ+:2015,JFJ+:2012}. Quipper offers a monad that supports ``lifting'' of quantum circuits from classical functions.
One of the key differences between the Quipper approach and our approach is that we sacrifice the concept of linear logic which is a crucial concept underlying Quipper and Selinger and Valiron's earlier work on the quantum lambda calculus \cite{SV:2009}.
Actually, we take a step in the opposite direction: the fundamental reason for our space saving potential is by allowing mutable variables and in-place updates of variables.
However, there are two prices to be paid for having this advantage: first, we do not have a linear type system anymore.
We believe that this is not a dramatic disadvantage in practice (and indeed we are not aware of any quantum compilation system that would have implemented type checks for {\em linearity} of the type system anyways as this leads to conflicts with the host language): the main advantage of linear types is that automatically consistency with regards to non-cloning is ensured.
As we focus on subroutines that are entirely classical/reversible this problem does not present itself.
Second, we now have to make sure that bits that there is a way to uncompute bits that have been used deep inside a computation that might have involved mutable {\em and} immutable variables.
If for each newly computed result a fresh ancilla is used, this task is trivial: the ancilla still holds the computed value and in order to uncompute another value based on the ancilla value, the result is still there.
In our case, it might have happened that the data in the ancilla itself might have been overwritten!
In this case there must be a clean way to track back the data in order to be able to recompute it.
To this end, we introduce a data structure that we call MDD; this stands for ''mutable data dependency graph'' and allows to track precisely this information.

There are also prior approaches toward higher-level synthesis based on reversible programming languages such as Janus and various extensions \cite{YG:2007,Thomsen:2012,Perumalla:2014}. These are languages allows only such programs to be expressed that are invertible in the sense that it is possible to recover the input back from the output. Also flow-chart like languages can be defined as a refinement of these reversible languages \cite{YAG:2008a,YAG:2008b} which allow a more fine-grained notion of reversibility. These languages constrain the programmer in a certain sense by forcing him or her to express the program already in a reversible way. This makes for instance expressing a function that is believed to be a one-way permutation $\pi : \{0,1\}^n \rightarrow \{0,1\}^n$ difficult, whereas in our approach ideally it should be easy to express $\pi$, provided it has an efficient circuit. The compiler should then be able to find an efficient circuit implementing $(x,y) \mapsto (x, y \oplus \pi(x))$. Tasks like these are possible in Janus and flow-chart like languages, however, they are somewhat unnatural. 

Finally, we would like to point out that there is a well-established theory of reversible synthesis for Boolean functions, i.e., the task of implementing an (irreversible) function that is given by its truth table by means of a reversible circuit. This typically entails methods that are quite efficient for a bounded number of bits \cite{MMD:2003,MMD:2007,MBC:2008,WD:2010,SSP:2013,LJ:2014}, techniques that achieve reversible synthesis by using the quantum domain \cite{MS:2011}, template based rewriting techniques that enable peep-holing methods to be applied \cite{MDM:2005} and even a complete classification of the optimal reversible circuits for small number of bits \cite{GM:2012}. 
See also \cite{SM:2013} for a survey. All these techniques aim at an understanding of reversibility at the lower level. In the present paper we aim at approaching reversible synthesis starting from high-level descriptions in a functional programming language. In this sense our approach is in principle scalable, whereas truth-table based methods can only be applied to a small number of bits. However, our approach leverages the techniques developed for truth table based synthesis as the atomic cases in which the compiler, after suitable decompositions, falls back to a known method once the given program has been decomposed into a collection of functions that act on a small number of bits.

\mypar{Our Contributions} Our main innovation is an algorithm to compute a data structure that we call the mutable data dependency graph (MDD) from a program.
This data structure tracks the data flow during a classical, irreversible computation and allows to identify parts of the data flow where information can be overwritten as well as other parts where information can be uncomputed early as it is no longer needed.
These two techniques of overwrite, which are implemented using so-called in-place operations, and early cleanup, for which we use a strategy that can be interpreted as a particular pebble game that is played on the nodes of the data flow graph, constitute the main innovation of the present work.
We study different embodiments of the cleanup strategy, depending on how aggressively the cleanup is implemented.
We implemented a compiler, which we call~\REVS, that can take any program from a language that is a subset of all valid F$\#$ programs and that can turn such a program into a corresponding  reversible networks over the Toffoli gate set.
The~\REVS{} compiler can be interfaced with the \LIQUID{} quantum programming language in a natural way as the output Toffoli network of the~\REVS{} compilation can be directly imported as an internal representation into \LIQUID{} and be used as part of another quantum computation. Furthermore, the simulation and rendering capabilities of \LIQUID{} can be used to further process the~\REVS{} output. Also, we implemented a simple, light-weight simulator for Toffoli networks as a part of~\REVS{} that can be used for testing purposes, e.g., on a set of random input/output pairs. 

A further contribution of our work is to demonstrate that higher-level reversible synthesis can be done in a way that is much more space efficient than using the Bennett method which essentially introduces additional ancillas per each operation used in the irreversible implementation. In our example implementations of arithmetic operations such as integer addition and multiplication, as well as hash functions, such as SHA-2 and MD5, we typically saw space savings of our method over the Bennett method of a factor of $4$ or more.

%
%

\section{Reversible computing}

It is known that the Toffoli gate \cite{NC:2000} which maps $(x,y,z) \mapsto (x,y,z\oplus xy)$ is universal for reversible computation. More precisely, it is easy to see that the group generated by all Toffoli, CNOT, and NOT gates on $n\geq 1$ bits is isomrphic to the alternating group $A_{2^n}$ of order $(2^n)!/2$. Hence, as $\mathbf{1} \otimes \pi = (\pi,\pi)$ is even for any permutation $\pi$, at the expense of at most 1 aditional bit any permuation can be implemented in terms of Toffoli gates. 

The number of Toffoli gates used in the implementation of a permutation $\pi$ is the basic measure of complexity that we use in this paper. This is motivated by universal fault-tolerant quantum computing where the cost of Clifford gates such as CNOT and NOT can usually be neglected, whereas the Toffoli gate has a substantial cost. In particular we ignore the cost of CNOT and NOT gates. It should be noted that other cost metrics such as multiple-controlled Toffoli gates have been studied in the literature which in turn can be related in terms of Toffoli cost, however, in our opinion counting of Toffoli gates is well-motivated due to the connection to fault-tolerant computing, whereas gate counts in terms of multiply controlled gates can hide a significant cost factor.

Bennett's work on reversible Turing
machines \cite{Bennett:73,Bennett:89} showed that it is possible to relate the cost of an irreversible circuit implementation of a function with that of a reversible circuit implementation. More precisely, if $x \mapsto f(x)$ denoting a Boolean function on $n$ bits that can be implemented with $K$ gates over $\{{\rm NOT}, {\rm AND}\}$, then the reversible function $(x,y) \mapsto (x,y\oplus f(x))$ can be implemented with at most $2K+n$ gates over the Toffoli gate set. The basic idea being to replace all AND gates with Toffoli gates, then perform the computation, copy out the result, and undo the computation, thereby cleaning up all bits. In a nutshell, the present paper tries to improve on the space-complexity of Bennett's strategy which is large as it scales with the circuit size of the irreversible circuit. 

To mitigate this, space-time trade-offs have been investigated via so-called reversible pebble games \cite{Bennett:89}. In general, pebble games are useful for studying programming languages and compiler construction, in particular for data dependency analysis. Data flow dependencies can be modeled by directed acyclic graphs \cite{ALS+:2007}. A pebble game played on this graph can then be used to model register allocation for the given data flow graph, and to explore time-space tradeoffs. The thesis \cite{Chan:2013} provides a summary of pebble games and a comparison between various notions of traditional pebble games that do not have to meet requirements of reversibility and Bennett's reversible pebble game. 

An important special case of reversible pebble games are games played on a directed line as they correspond to function with sequential dependency on intermediate results. We denote the time and space resources needed for the irreversible function by $T$ and $S$ and those for the reversible function by $T_{rev}$ and $S_{rev}$. Bennett's original paper \cite{Bennett:73} asserts that $T_{rev} = O(T)$ and $S_{rev} = O(ST)$ is possible, however, in terms of the required space this has significant downside as in the worst case $T$ might scale as $T=O(2^S)$. For Bennett's reversible pebble game, i.e., the pebble played on line, it is known \cite{BTV:2001} that if the resources are constrained to $k$ pebbles, then there are the following upper bounds: $T_{rev} = O(S 3^k 2^{O(T/2^k)})$ and $S_{rev}=O(kS)$. As limiting special cases \cite{BTV:2001} we obtain for $k=O(1)$ that $T_{rev}=O(S 2^{O{T}})$ and $S_{rev} = O(S)$ which is the Lange-McKenzie-Tapp strategy \cite{LMT:2000} that has exponential time cost but stays linear in the required space. Another special case is for $k=\log{T}$ for which we obtain Bennett's space/time tradeoff \cite{Bennett:89} which runs in $T_{rev}=O(S 3^{\log_2{T}}) =O(S T^{\log_2{3}})$ and $S_{rev} = O(S \log{T}) = O(S^2)$. One can further improve the time dependency \cite{Bennett:89,LS:90} to $T_{rev}= O(T^{1+\varepsilon}/S^\varepsilon)$ for any given $\varepsilon>0$ and space dependency $S_{rev}=O(S(1+\ln(T/S)))$. To the best of our knowledge, only little is known about the ultimate limit of these space-time tradeoffs, see however \cite{FA:2001} in which the authors provide an oracle separation between classical and reversible simulations that run in $S_{rev}=O(S)$ space and $T_{rev}=O(T)$ time. While for the directed line, the problem of finding the optimal pebbling strategies for a given space can be solved in practice quite well using dynamical programming \cite{Knill:95}. For general graphs, finding the optimal strategy is PSPACE complete \cite{Chan:2013}, i.e., it is unlikely to be solvable efficiently. 

A simple and easy to implement version of the 1D pebble games is the ``incremental'' pebble game.
In this pebble game we simply add pebbles until we run out.
We then remove as many pebbles as we can starting at the point where we ran out and use them to continue the computation.
We leave behind a pebble each time we do this.
It is easy to see that for some amount of pebbles $n$ we can pebble a distance that scales as $n + (n-1) + (n-2) + \dotsb + 1= \frac{n(n+1)}{2}$. 
And since we will pebble/unpebble a given node a maximum of 4 times (twice in the forward computation and twice again during clean-up), the total amount of computations is worst case is $4N$ where $N$ is the number of irreversible operations. In other words, this means that we get a scaling of $T_{rev} = O(T)$ and $S_{rev}=O(S \sqrt{T})$ for this incremental pebble strategy. 

\begin{figure}[htb]
\centering
\begin{tabular}{c@{\qquad}c@{\qquad}c}
\includegraphics[height=0.6cm]{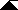} &
\includegraphics[height=1.2cm]{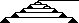} &
\includegraphics[height=1.6cm]{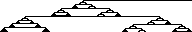} \\
(a) & (b) & (c)
\end{tabular}
\caption{\label{fig:pebblegames}
Visualization of three different pebble strategies that succeed in computing a pebble game on a linear graph, but use different strategies. Time is displayed from left to right, qubits are displayed from bottom to top. The strategy shown in (a) corresponds to Bennett's compute-copy-uncompute method \cite{Bennett:73} where the time cost is linear but the space cost is $O(S T)$, (b) is the incremental strategy mentioned in the text for which the time cost is linear but the space cost is $O(S \sqrt{T})$. The strategy shown in (c) corresponds to the Lange-McKenzie-Tapp method \cite{LMT:2000} that takes exponential time but has only a space cost of $O(S)$.
}
\end{figure}

In Figure \ref{fig:pebblegames} different strategies to clean up a computation on a linear 1D graph are visualized. Time is displayed from left to right, computational nodes are displayed from bottom to top---the output of the topmost gate being the final output---, and pebbles are denoted by black pixels. The state of the game after $i$ time steps is precisely the vertical slice obtained at location $i$. Bennett's original strategy corresponds to (a) in which case shown in the figure a sequence of $10$ gates is pebbled using $10$ pebbles and $19$ time steps. An extremely space-efficient strategy is shown in (c). Here in a fashion that resembles a fractal, a strategy is implemented that pebbles $32$ gates using only $6$ pebbles but which needs $193$ time steps. In (b), a possible middle ground is shown, namely an incremental heuristic that first uses up as many pebbles as possible, then aggressively cleans up all bits except for the last bit, and the repeats the process until it ultimately runs out of pebbles. Here $24$ gates are pebbled using $5$ pebbles.

\begin{figure}[htb]
\centering
{\hspace*{-0.3cm}\includegraphics[width=0.7\columnwidth]
{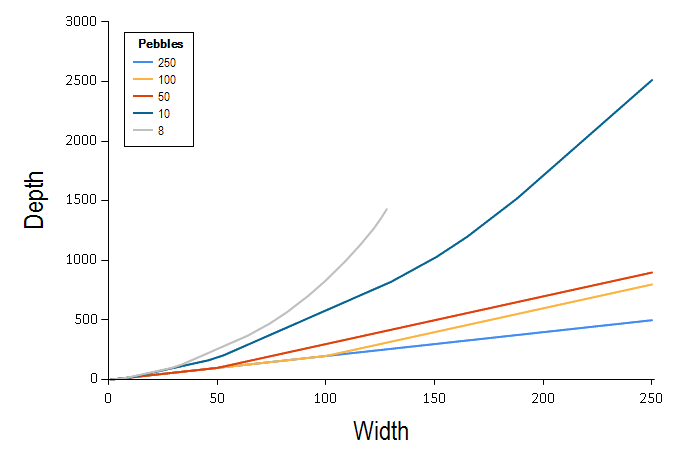}}
\caption{\label{fig:pebbleplot}
Comparison of the asymptotic scaling of different pebbling strategies. The lower curve is Bennett's original strategy \cite{Bennett:73} in which the number of pebbles scales with the number of time steps. The upper curve is the space-efficient but exponential-time Lange-McKenzie-Tapp strategy \cite{LMT:2000}. The other curves show the graceful degradation of the time complexity, here measured as circuit depth, with the reduction of the number of available pebbles.}
\end{figure}

Knill \cite{Knill:95} showed that for 1D pebble games the cost for the optimal pebbling strategy for any fixed number of available pebbles can be expressed recursively. Using dynamic programming, we used this to implement a search for optimal pebbling strategies on 1D graphs for a variety of different space constraints. The findings are summarized in the plot shown in Figure \ref{fig:pebbleplot}. In there the uppermost curve corresponds to the pebble games shown in Fig.~\ref{fig:pebblegames}(c). The lowermost curve is the original Bennett strategy shown in Fig.~\ref{fig:pebblegames}(a) in which the number of ancillas scales linearly with the depth of the circuit, Shown are also various other strategies that are optimal for given space resource constraints. To summarize our findings of this dynamic programming search for 1D pebble game under the given space constraints is that a significant reduction of the number of pebbles by a factor of $4$ or more leads to an almost negligible increase in the length of the reversible computation. The extremely parsimonious strategies that achieve a pebbling of a 1D line graph of size $n$ with $O(\log(n))$ pebbles do not seem advisable from a practical standpoint as the circuit size increases dramatically. We believe that in practice a middle ground will be best, in particular constant reduction in terms of the number of pebbles which then lead to only small increase in circuit size and compilation time.

While for the directed line, the problem of finding the optimal pebbling strategies for a given space can be solved in practice quite well using dynamical programming \cite{Knill:95}, for general graphs, finding the optimal strategy is PSPACE complete \cite{Chan:2013}, i.e., it is unlikely to be solvable efficiently. In our paper we employ heuristics to find pebbling strategies that work well in practice and often lead to better space complexity than Bennett's original method, which is the method implemented e.g. in Quipper \cite{GLR+:2013b}.  We make no claims about optimality of our strategies but present empirical evidence for space reductions of $S_{rev}$ of around $70\%$ at only a moderate increase in $T_{rev}$.  

The cleanup strategies we develop in the parts of the paper can be thought of as pebble games that are played on finite graphs, namely the dependency graphs of the functions that are computed by the program. In our paper we employ heuristics to find pebbling strategies that work well in practice and often lead to better space complexity than Bennett's original method, which is the method implemented e.g. in Quipper \cite{GLR+:2013b}.  We make no claims about optimality of our strategies but present empirical evidence for space reductions of $S_{rev}$ over the original Bennett strategy \cite{Bennett:73} of around a factor of $4$ or more at only a moderate increase in $T_{rev}$.

\subsection{Boolean functions}\label{sec:boolean}
An important special case of programs that need to be turned into reversible circuits are Boolean functions.

Boolean functions are used as the primitives in our implementation.
All supported boolean operators are converted into AND/XOR functions and grouped into boolean expressions.
The expressions are then converted into Toffoli/CNOT circuits while attempting to minimize ancilla use.
This is done by combining operations into expressions of the type:
\begin{lstlisting}[language=FSharp]
type BoolExp =
	| BVar of int
	| BAnd of BoolExp list
	| BXor of BoolExp list
    | BNot of BoolExp
\end{lstlisting}
Expressions are always given a target to be evaluated onto. Each BXor term can be constructed by evaluating each term then adding a CNOT from each of them to a given target. Each BAnd term can be constructed using a multiple control Toffoli decomposition targeted again to the given target. This means ancilla usage is limited to the bits required to preform all of the AND operations in the expression. In general, the output of a Boolean function is of the form $y \oplus e$ where $y$ is the target bit and $e$ is the given Boolean expression. Note that this allows the evaluation of $e$ by initializing $y=0$ as $0 \oplus e = e$.

It is possible to do further optimization by factoring the expression in an attempt to remove and operations.
For example $ab\oplus ac \oplus bc$ can be factored as $a(b \oplus c) \oplus bc$ so that it uses 2 and operations rather then 3.
Currently there is no automated factoring but if the expression is written in a factored form by the programmer it will result in better circuit generation.

Finally, we remark that some operations can be performed in-place in a reversible fashion. An operation is in-place if it modifies data without creating any ancilla.
For example the CNOT gate preforms the operation $(a,b)\mapsto (a,a\oplus b)$.
So if we wish to preform the operation $a\oplus b$ and do not require $b$ later in the circuit we do not need an additional ancilla to store the output.
Other examples for in-place operations are integer adders. More generally, it is easy to see that any function $f$ that a) implements a permutation, b) for which an efficient circuit is known and c) for which $f^{-1}$ can also be implemented efficiently, gives rise to an efficient in-place reversible circuit that maps $x \mapsto f(x)$.

%
%

\section{REVS: a compiler for reversible circuits}

\subsection{Dependency analysis}

Analyzing the dependencies between the instructions in a basic function, between functions, and between larger units of code is a fundamental topic in compiler design \cite{ALS+:2007,Muchnick:97}. Typically, dependency analysis consists of identifying basic units of codes and to identify them with nodes in a directed acyclic graph (DAG). The directed edges in the graph are the dependencies between the basic units, i.e., anything that might constrain the execution order for instance control dependencies that arise from the control flow in the program which in turn can be for instance branching that happen conditional on the value of a variable or, more simply, the causal dependencies that arise from one unit having to wait for the output of another unit before the computation can proceed.

In our case, the dependency graph is generated in a two step process. First, the F$\#$ compiler is invoked to generate an abstract syntax tree (AST) for the input program. This is done using the mechanism of reflection for which F$\#$ offers support in the form of so-called quotations \cite{SGC:2012}. Quotations have a simple syntax by surrounding expressions for which an abstract syntax expression is to be constructed with {\blue \texttt{<@ ... @>}}. F$\#$ quotations are types which implies that much of the type information present in the program as well as the expression based nature can be leveraged. In practice this means that the AST will already be represented in a form that can then be easily dispatched over by using {\blue \texttt{match}}
statement for the various constructors that might be used. Second, we then use so-called active patterns in our match statements to further aid with the process of walking the AST and turning it into an internal representation that represents the dependency graph of the program.

The nodes of this graph captures the control flow and data dependencies between expressions, but also identifies which blocks can be computed by in-place operations and which blocks have to be computed by out-of-place operations. Because of this latter feature is related to which elements of the dependency graph are mutable and which are not, we call this data structure the Mutable Data Dependency graph or MDD. Which parts of the code can be computed by in-in-place operation is inferred by which variables are labeled in F$\#$ as {\blue \texttt{mutable}} together with the external knowledge about whether for an expression involving these variables an in-place implementation is actually known. An example for the latter is the addition operation for which we can chose either an in-place implementation $(a,b) \mapsto (a,a+b)$ or an out-of-place implementation $(a,b,0) \mapsto (a,b,a+b)$.

The nodes of the MDD correspond to inputs, computations, initialized and cleaned-up bits. Inputs nodes can correspond to individual variables but also to entire arrays which are also represented as a single node and treated atomically. Computation nodes correspond to any expression that occurs in the program and that manipulates the data. Initialized and cleaned-up bits correspond to bits that are part of the computation and which can be used either as ancillas or to hold the actual final output of the computation. Initialization implies that those qubits are in the logical state $0$ and the cleaned-up state means these bits are known to be returned back in the state $0$.

\begin{algorithm}[hbt]
\caption{{\textsc{MDD}} Computes mutable data dependency graph.}
  \begin{algorithmic}[1]
      \Require {AST : The AST of a function to be compiled}
      \Procedure{resolveAST}{AST,G}
        \If {Root of AST is an operation}
          \For{input \textbf{in} inputs(AST)}
            \State inputIndex , G $\leftarrow$ ResolveAST(input,G)
            \State inputIndices   $\leftarrow$ inputIndex :: inputsIndices
          \EndFor
          \State newNode.type $\leftarrow$ OpType(head(AST))
        \State newNode.inputs $\leftarrow$ addInputArrows(inputIndices)
          \State G $\leftarrow$ AddNode(newNode)
          \State \textbf{return}  getIndex(newNode) , G
        \Else
          \State \textbf{return} getVarIndex(head(AST)) , G
        \EndIf
      \EndProcedure
      \State G $\leftarrow$ Add nodes for all inputs
      \State resolveAST(AST,G)
  \end{algorithmic}\label{alg:depgraph}
\end{algorithm}

The directed edges in a MDD come in two different kinds of flavors: data dependencies and mutations.
Data dependencies are denoted by dashed arrows and  represent any data dependency that one expression might have in relation to any other expression.
Mutations are denoted by bold arrows and represent parts of the program that are changed during the computation.
By tracking the flow of the mutations one can then ultimately determine the scheduling of the expressions onto reversible operations and re-use a pool of available ancillas which helps to reduce the overall space requirements of the computation, in some cases even drastically so.
A high level description of the algorithm that computes the MDD from the AST produced by the F$\#$ compiler is given in Algorithm~\ref{alg:depgraph}.

When resolving the AST of a function, each node will either be another function or an input variable. If the node is a function, Algorithm \ref{alg:depgraph} recursively computes the AST for all of the function inputs adding the results to the graph. Upon doing so, we use the index numbers of these results as the inputs for the operation and then add the operation to the graph. If the node is a variable, the algorithm looks up its name in a map of currently defined variables and returns an index to its node. The type of the operation determines which arrows will be solid input arrows and which will be data dependencies, i.e., controls. In the example figures given in this paper, the paths from inputs to outputs that indicate modifications are drawn using bold arrows, whereas controls are shown as dashed arrows.
\nix{In the F\# implementation assignment ($\leftarrow$) is used when we want an input arrow. This means we need not worry about later statements referencing the earlier value of the variable.}
As the algorithm visits each node in the AST and does a modification of the graph that involves only a constant number of elements, it is clear that the overall runtime of Algorithm \ref{alg:depgraph} is $O(n)$, where $n$ is the number of nodes in the AST.

A simple example for this translation of a program into an MDD is given by the following program:

\begin{lstlisting}[language=FSharp]
	let f a b = a && b
\end{lstlisting}

In this program $f$ is simply the AND function of two inputs $a$ and $b$.
The MDD corresponding to the program is shown in Figure \ref{fig:simpleGraph}.

\begin{figure}[htb]
\centering
\includegraphics[width=0.3\columnwidth]{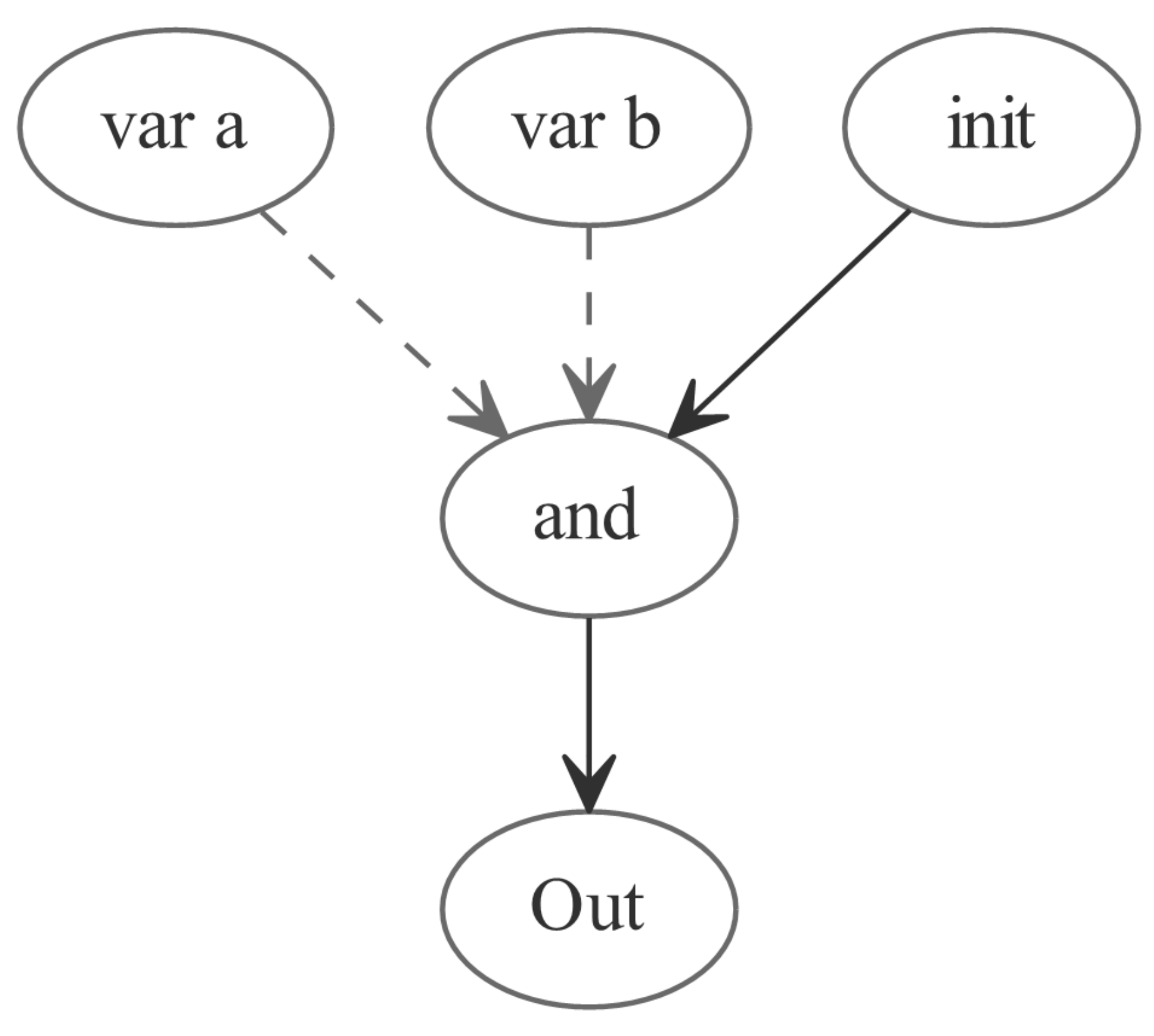}
\caption{\label{fig:simpleGraph}
MDD for $f(a,b)=a \land b$
}
\end{figure}

There are two input nodes in the MDD for $f$ labeled as ``var $a$'' and ``var $b$'' in the figure. Those nodes are immutable. Furthermore, there is one node initialized in $0$, denoted as ``init'' and one node which will contain the final output, denoted as ``Out''. Data dependencies are present in the node for the AND operation, meaning that a node for the AND operation is created and (dashed) input arrows pointing from the variables $a$ and $b$ are added.
The node ``init'' is used to hold the result. It points to the operation with a solid modification arrow. Finally an output node is added showing which value is being treated as the output of the function.
The final code emission by the compiler in this concrete case will use a library for Boolean expressions which in this concrete case then is invoked when mapping the MDD to a reversible circuit. The circuit corresponding to this particular program is just a single Toffoli gate with control qubits $a$ and $b$ and one target qubit.

\begin{figure}[htb]
\centering
\begin{lstlisting}[language=FSharp]
let xor4 (a:bool array) (b:bool array) =
	let c = Array.zeroCreate 4
	for i in 0 .. 3 do
			c.[i] <- a.[i] <> b.[i]
	c
let and4 (a:bool array) (b:bool array) =
	let d = Array.zeroCreate 4
	for i in 0 .. 3 do
		d.[i] <- a.[i] && b.[i]
	d
let mutable a = Array.zeroCreate 4
let b = Array.zeroCreate 4
let c = Array.zeroCreate 4
a <- xor4 a b
and4 a c
\end{lstlisting}
\caption{\label{fig:fsharpArray}
Simple F$\#$ snippet using an arrays and in place operations. The mutable data dependency (MDD) graph corresponding to the \texttt{and4} function is given in Figure \ref{fig:arrayGraph}
}
\end{figure}

A slightly more involved example is given by the code shown in Figure \ref{fig:fsharpArray}.
Here there are new several elements as compared to the simple Boolean example above that make the MDD construction slightly more non-trivial.
First, a number of arrays are used to store data in a way that allows for easy access and indexing.
Note that in F$\#$ the type {\blue \texttt{array}} is inherited from the .NET array type and by definition is a mutable type.
This information is used when the MDD for the program is constructed as the~\REVS{} compiler knows that in principle the values in the array can be updated and overwritten.
Whether this can actually be leveraged when compiling a reversible circuit will of course depend on other factors as well, namely whether the parts of the data that is invoked in assignments (denoted by {\blue \texttt{<-}}) is used at a later stage in the program, in which case the data might have to be recomputed.

Note further that there are basic control flow elements such as for-loops and function calls and Boolean connectives, namely the AND function that we already met in the previous example, and the XOR function, denoted by  {\blue \texttt{<>}}. The MDD corresponding to the main function {\blue and4} is shown in Figure \ref{fig:arrayGraph}.

\begin{figure}[htb]
\centering
\includegraphics[width=0.25\columnwidth]{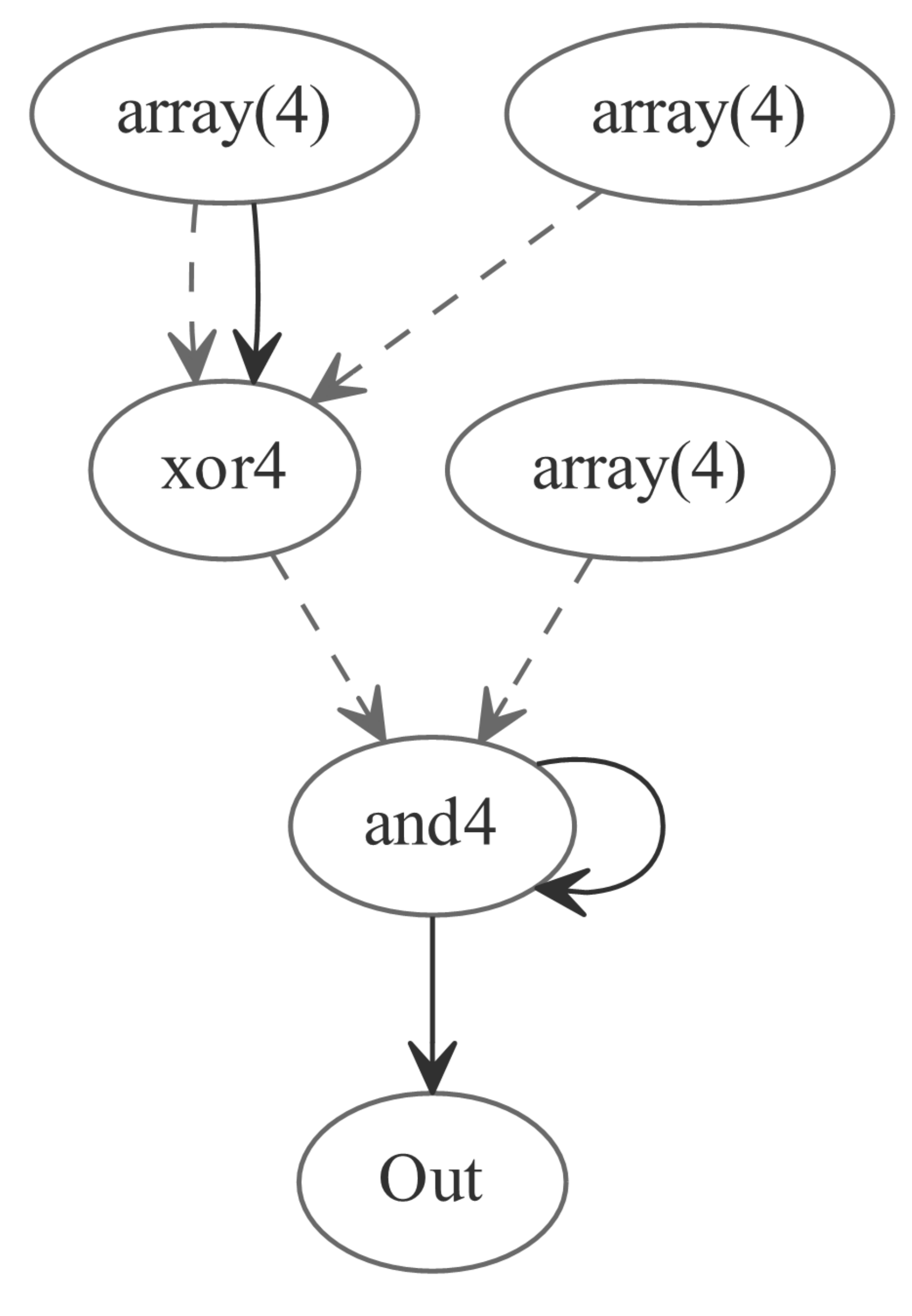}
\caption{\label{fig:arrayGraph}
MDD of the main function \texttt{and4} from Figure \ref{fig:fsharpArray}
}
\end{figure}

\subsection{Clean-up strategies}
If a node has no outgoing modification arrows and all operations pointed to by its dependency arrows have been completed it is no longer needed by the computation and may be cleaned.

\subsubsection{Eager Clean-up}
With eager clean-up we start at the end of the graph and work our way back in topological order.
When we find a node ($A$) which does not have an outgoing modification arrow we first find the node furthest along in topological order which depends on it ($B$).
We then consider all inputs in the modification path of $A$.
If any of the inputs have outgoing arrows modification arrows pointing levels previous to $B$ we may not clean the bit eagerly since its inputs are no longer available.
If the inputs do not have modification arrows pointing at levels previous to $B$ we can immediately clean it up but reversing all operations along its modification path.

For example see Figure \ref{fig:OrGraph} generated by the $f$ function in the code:
\begin{lstlisting}[language=FSharp]
let f a b = a || b
let g a b = a && b
let h a b c d = f a b <> g c d
\end{lstlisting}
The values initialized to hold the results of $ab$ and $a\oplus b$ are no longer required after they are used to calculate the final result.
The compiler notices that the original values used to produce them are still available and uncomputes the extra ancilla.
Notice in the circuit graph of this operation (Figure \ref{fig:fg}) that the freed ancilla can be reused in the other parts of the computation.

Ancilla currently in use are tracked during circuit generation using a heap data structure.
Whenever an ancilla is needed during the compilation from the MDD into a circuit, an identifier (implemented as a number) is taken off the heap and the bit matching that number is used.
After a bit has been cleaned up the corresponding identifier is pushed back onto the heap.
This allows ancilla to be reused and ensures that we use only the minimum indexed ancllia so we can avoid allocating unneeded space.

\begin{figure}[ht]
  \centering
  \begin{subfigure}[b]{0.45\textwidth}
    \includegraphics[width=\textwidth]{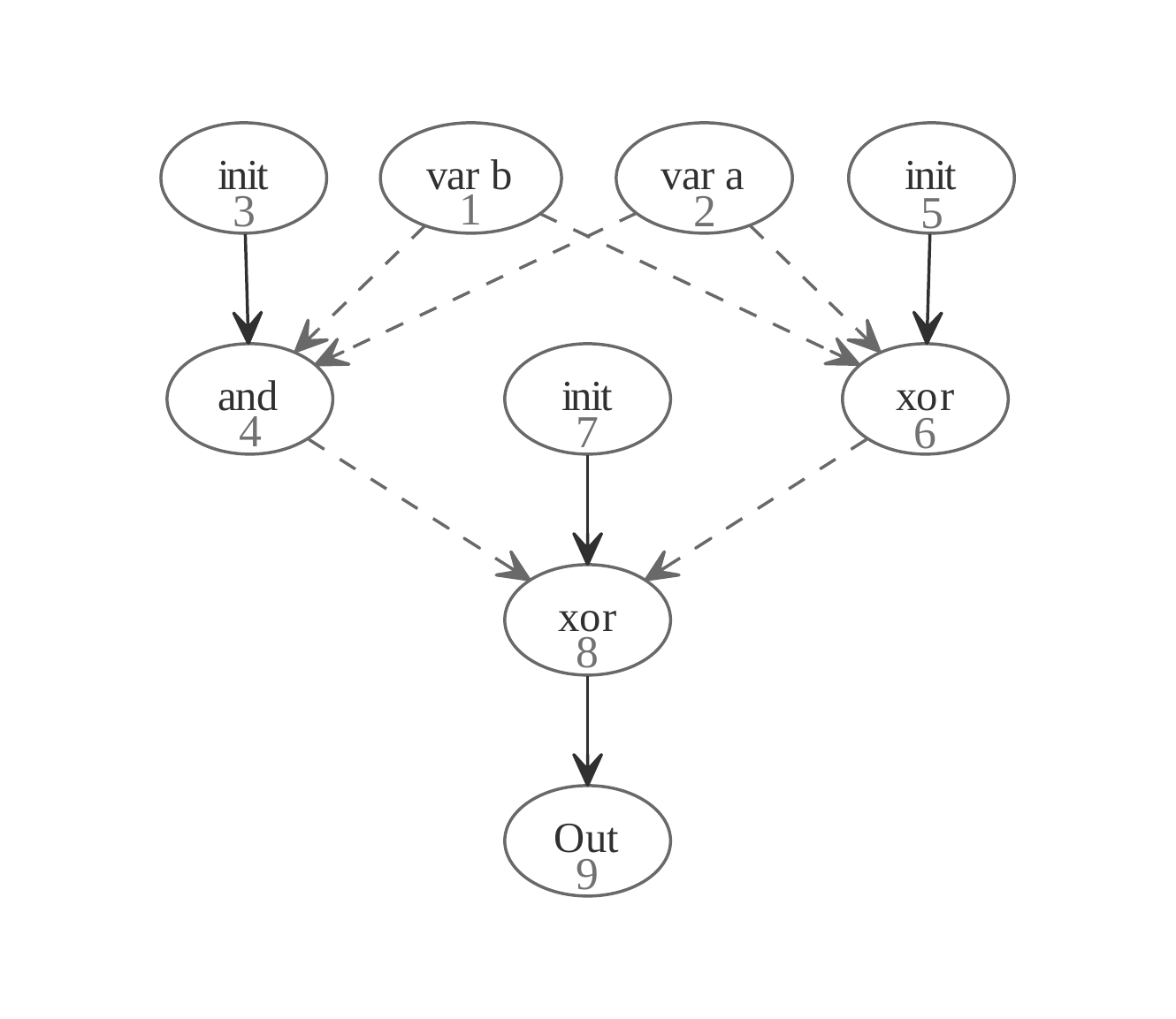}
    \caption{Before cleanup}
  \end{subfigure}
  \begin{subfigure}[b]{0.45\textwidth}
    \includegraphics[width=\textwidth]{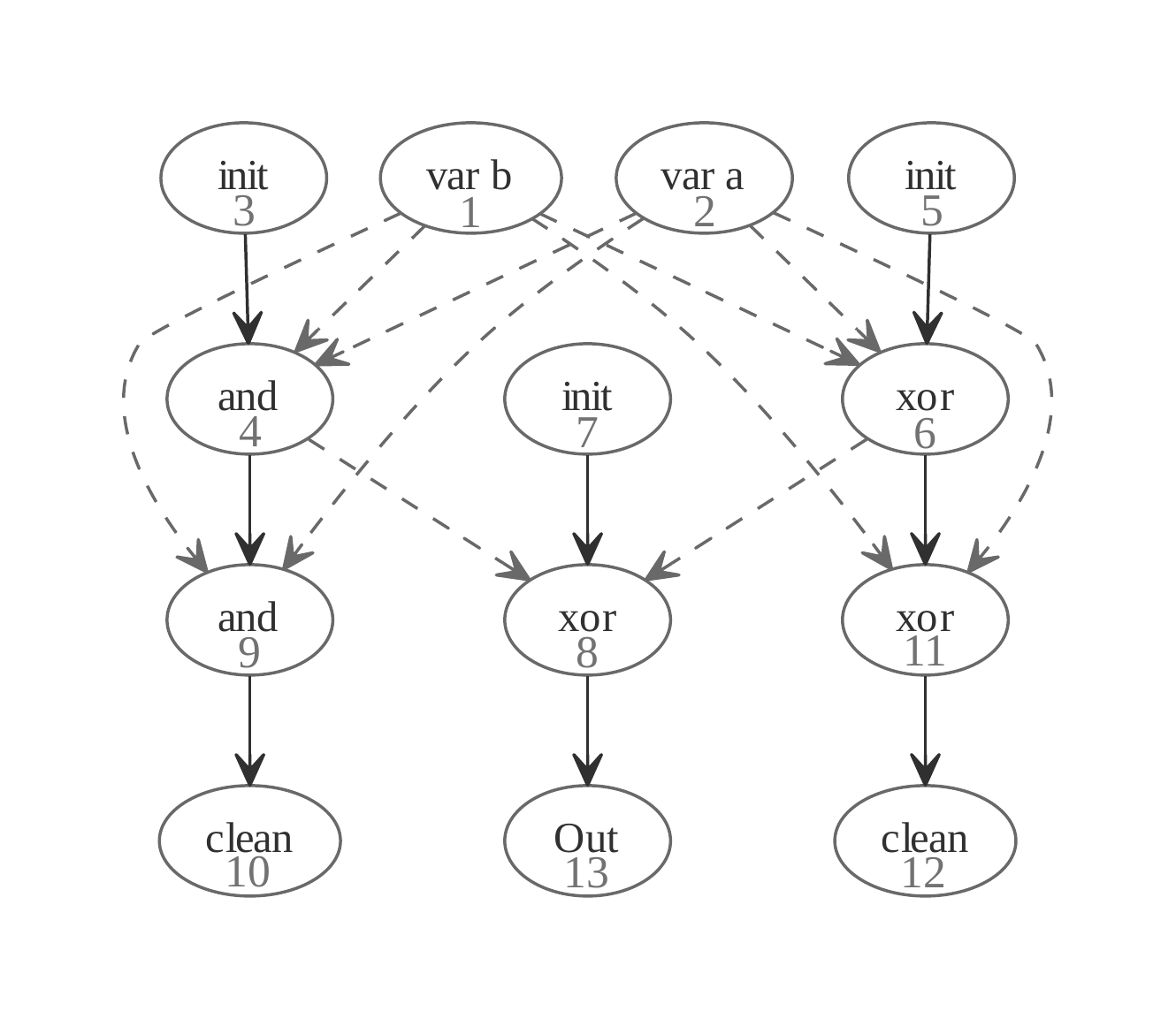}
    \caption{After cleanup}
  \end{subfigure}
  \caption{\label{fig:OrGraph}Dependency graph for $a||b$ which is represented as $ab \oplus a \oplus b$.
            The nodes are numbered with a possible topological order.
            Starting from the bottom we see that the lowest index which is not an input (\texttt{var}) or output is \texttt{xor} at 6.
            The last dependent node of 6 is xor at 8 and the input nodes (\texttt{var a} and \texttt{var b}) have not been modified after 6.
            We are therefore able to insert nodes after node 8 to cleanup 6.
            The nodes we insert are numbered 11 and 12 in the final graph.
            They are initially numbered 9 and 10 but when we cleanup on the left side we again insert the nodes after node 8,
            moving the indices of the previously inserted nodes up by 2.}
\end{figure}

Several comments about this example are in order. First, note that like all the other examples in this paper, the shown F$\#$ program can be compiled and executed on a classical computer just like any other F$\#$ program. By putting quotations around the program and sending it to the~\REVS{} compiler, we can give another semantic interpretation of the same program, namely to map it to a reversible network.
As can be seen in Figure \ref{fig:OrGraph}, operations on the same level can be reordered or even preformed in parallel without changing the outcome of the computation. Finally, the Toffoli network emitted by the compiler at the back-end, i.e., in the circuit generation phase, is shown in Figure \ref{fig:fg}.

\begin{figure}[htb]
\centering
\includegraphics[width=0.7\columnwidth]{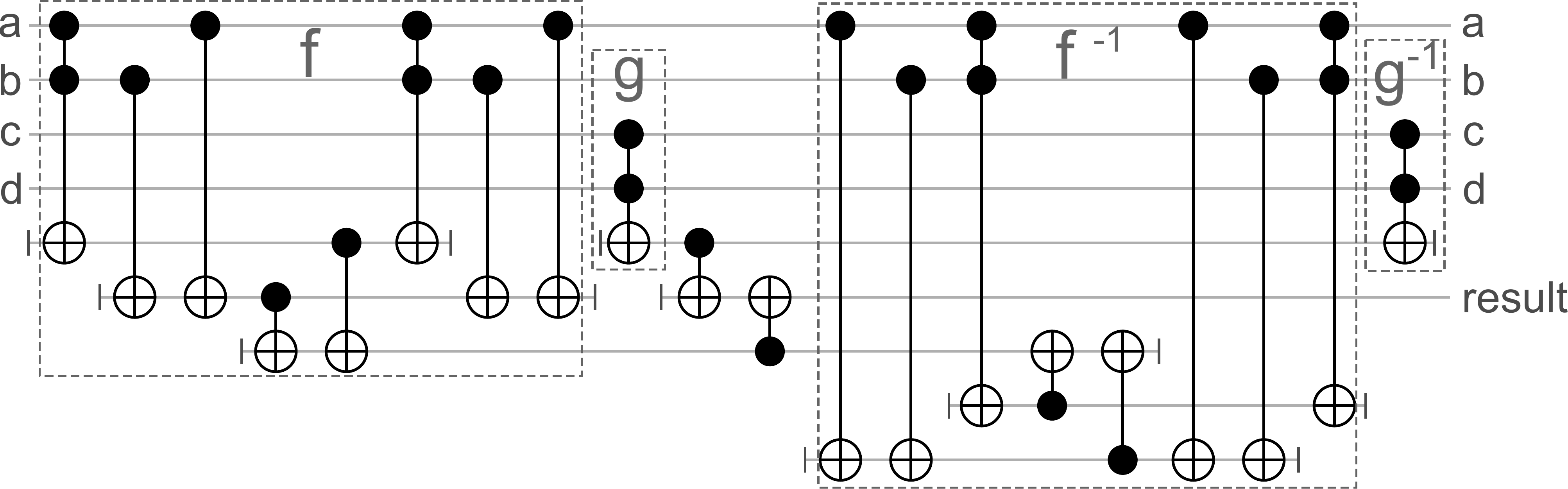}
\bigskip
\caption{\label{fig:fg} The resulting reversible network for computing the function $h(a,b,c,d) = f(a,b) \oplus g(c,d)$. The sub-circuit labeled with ``f'' is obtained from the MDD shown in Figure \ref{fig:OrGraph}. Similarly, a sub-circuit for $g$ is obtained from an MDD for $g$, which is an AND function. The final result is copied out into the ``result'' wire and then cleanup is performed by running the circuits for $f$ and $g$ in reverse. 
}
\end{figure}

\nix{

}

\begin{algorithm}[hbt]
	\caption{{\textsc{EAGER}} Performs eager clean-up of an MDD.}
	\begin{algorithmic}[1]
		\Require{An MDD $G$ in reverse topological order,  subroutines {LastDependentNode}, {ModificationPath}, {InputNodes}. }
		\State $i \leftarrow 0$
    \For{\textbf{each} node \textbf{in} G}
			\If{ modificationArrows node $= \varnothing$}
        \State dIndex $\leftarrow$ LastDependentNode of node in $G$
				\State path   $\leftarrow$ ModificationPath of node in $G$
        \State input  $\leftarrow$ InputNodes of path in G
        \If {None (modificationArrows input) $\geq$ dIndex}
          \State cleanUp $\leftarrow$ (Reverse path) ++ cleanNode
				\EndIf
		 \Else \State cleanUp $\leftarrow$ uncleanNode
     	\State $G$ $\leftarrow$ Insert cleanUp Into $G$ After dIndex
			\EndIf
		\EndFor\\
		\Return $G$
	\end{algorithmic}\label{alg:eager}
\end{algorithm}

It turns out that it is not always possible to do clean-up eagerly.
The basic reason for this is that the computation might result in the production of bits which are not needed in the future execution path of the circuit but which also cannot be easily cleaned up as they themselves were the result of another computation. A simple example is the following:
\begin{lstlisting}[language=FSharp]
let a = false
let mutable b = false
let c = (a&&b)
b <- b <> c
b
\end{lstlisting}

Here the variable $c$ is the result of a non-trivial computation and is then part of a subsequent computation that involves updating yet another variable $b$ in a mutable update path. The corresponding MDD is shown in Figure \ref{fig:NoEager}.

\begin{figure}[htb]
  \centering
  \begin{subfigure}[b]{0.25\textwidth}
    \includegraphics[width=\textwidth]{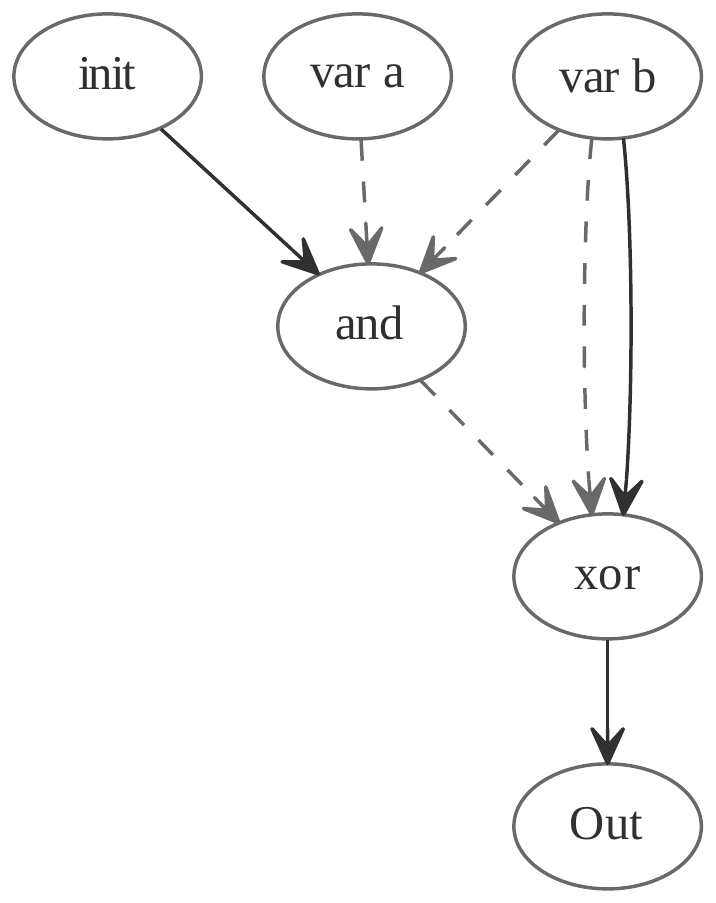}
    \caption{}
  \end{subfigure}
  \hfill
  \begin{subfigure}[b]{0.24\textwidth}
    \raisebox{1cm}{\includegraphics[width=\textwidth]{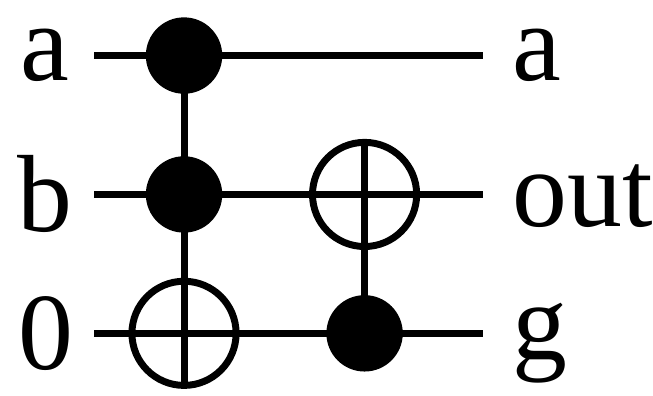}}
    \caption{Circuit for the graph in (a)}
  \end{subfigure}
  \hfill
  \begin{subfigure}[b]{0.39\textwidth}
    \raisebox{1cm}{\includegraphics[width=\textwidth]{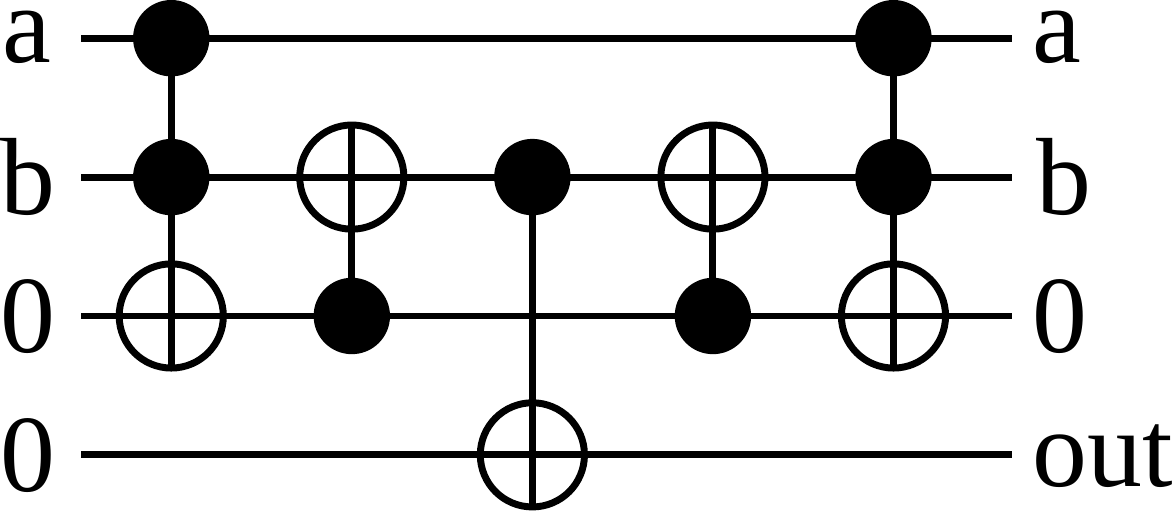}}
    \caption{Circuit after cleanup by copy\&reverse}
  \end{subfigure}

  \caption{\label{fig:NoEager}
  Example illustrating a case in which early cleanup is not possible.
  Here the last dependent node of the \texttt{and} operation is the \texttt{xor} operation but the input \texttt{var b} to the \texttt{and} operation has a modification arrow to \texttt{xor}.
  This means that eager cleanup will fail (see algorithm \ref{alg:eager}).
  }
\end{figure}

The graph shows that a value is initialized to store the result of the AND operation but cannot be cleared even though it is not longer needed in the computation.
The reason for this is that the node performing the XOR computation depends on it, indicated by the dashed arrow, and the mutable data $b$ is modified before there is an opportunity to clear it.
Situations like this can be identified in the MDD graph as finding (undirected) cycles involving at least one bold edge.
If a case like this occurs, the function will not be fully cleaned.
To resolve this we can copy out the result and reverse the function.

A pseudo-code implementation of the algorithm to perform eager cleanup is given in Algorithm \ref{alg:eager}. The algorithm takes a graph in reverse topological order and tries to find a node that doesn't modify any nodes that come after it in reverse topological order, so that this node can be safely cleaned up. As finding such a node involves checking all nodes that might influence its value along the modification path to which it belongs, which in itself might take a linear time $O(n)$ of checks, where $n$ is the input size as measured by the number of nodes in the initial MDD $G$, we obtain an overall worst-case running time of $O(n^2)$ for this algorithm.

The algorithm uses $3$ subroutines: LastDependentNode($G_i$,$G$) is defined to be the index of the last node in topologically a sorted graph $G$ which depends on node $G_i$. ModificationPath($G_i$,$G$) is defined as the path made up of the mutation arrows from initialization to $G_i$. Finally, InputNodes($p$,$G$) is defined as the set of all input nodes into a path $p$.

As discussed above, there are cases, in which no eager cleanup is possible, i.e., where the if statement in line 7 of the code does not apply. This happens if there are modification arrows greater then the index of the last dependent node in the input, i.e., the input has changed before eager cleanup can be done, i.e., eager cleanup is not possible (this is as in the example in Figure \ref{fig:NoEager}). In this case, the results must be copied out and the function reversed for a full cleanup. This is marked by the attribute ``unclean'' when the else branch in line 10 in Algorithm \ref{alg:eager} is executed. In this case, during the final circuit emission phase which involves walking the modified MDD $G$ that is returned by Algorithm \ref{alg:eager}, the unclean nodes are processed by copying out the result and cleaning up as in the Bennett strategy.

\subsubsection{Proof of correctness of the eager cleanup method}

We provide a proof of correctness of the \textsc{Eager} cleanup strategy introduced in the previous section. First, we introduce two notions of relationships that mutation paths that are part of an MDD can have, namely independence and inter-dependence. This naturally captures and generalizes the notion of conflicts shown in Figure \ref{fig:NoEager}. 

\begin{figure}[htb]
\centering
\begin{tabular}{c@{\qquad\qquad\qquad}c}
\includegraphics[height=3.5cm]{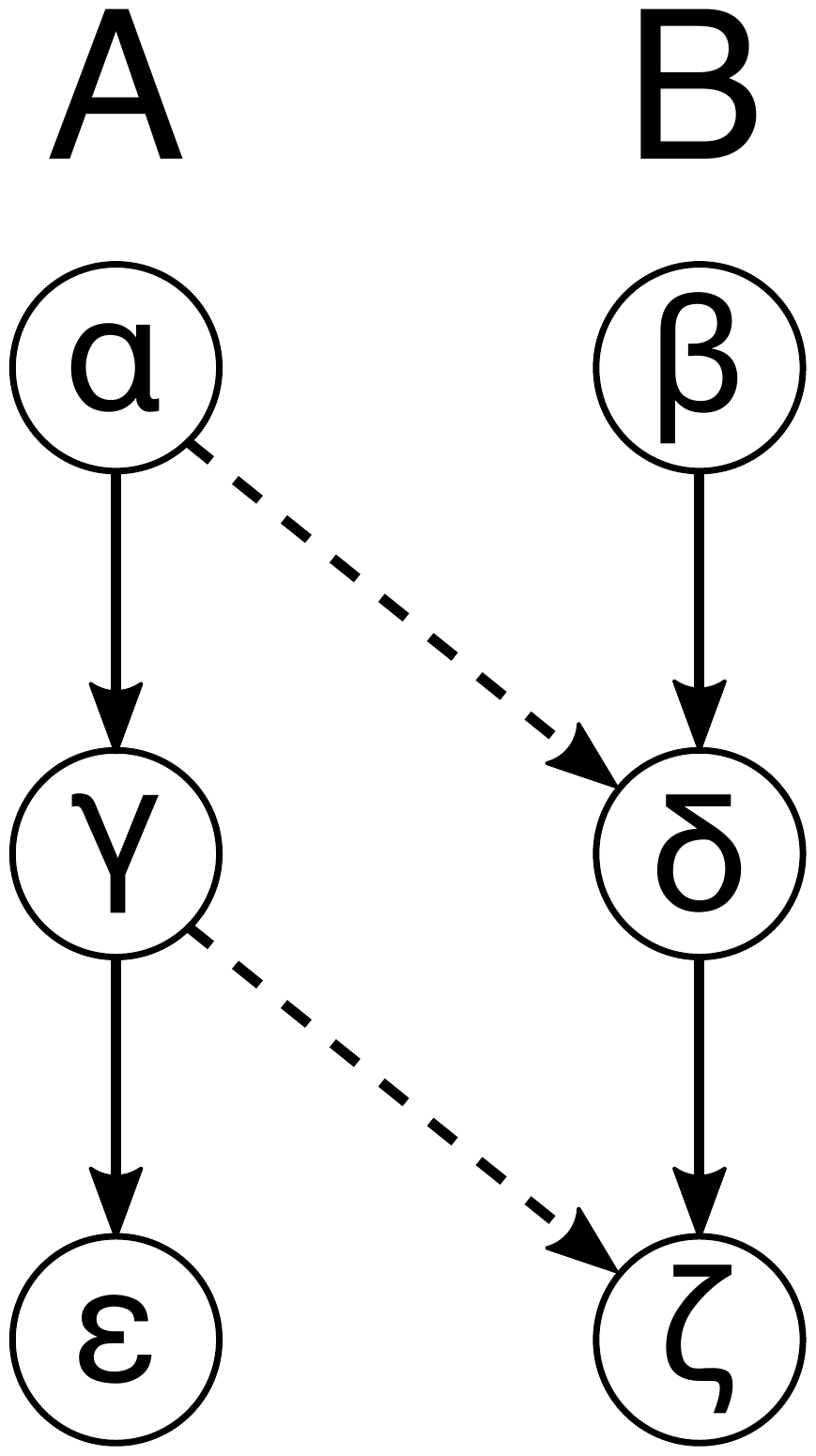} &
\includegraphics[height=3.5cm]{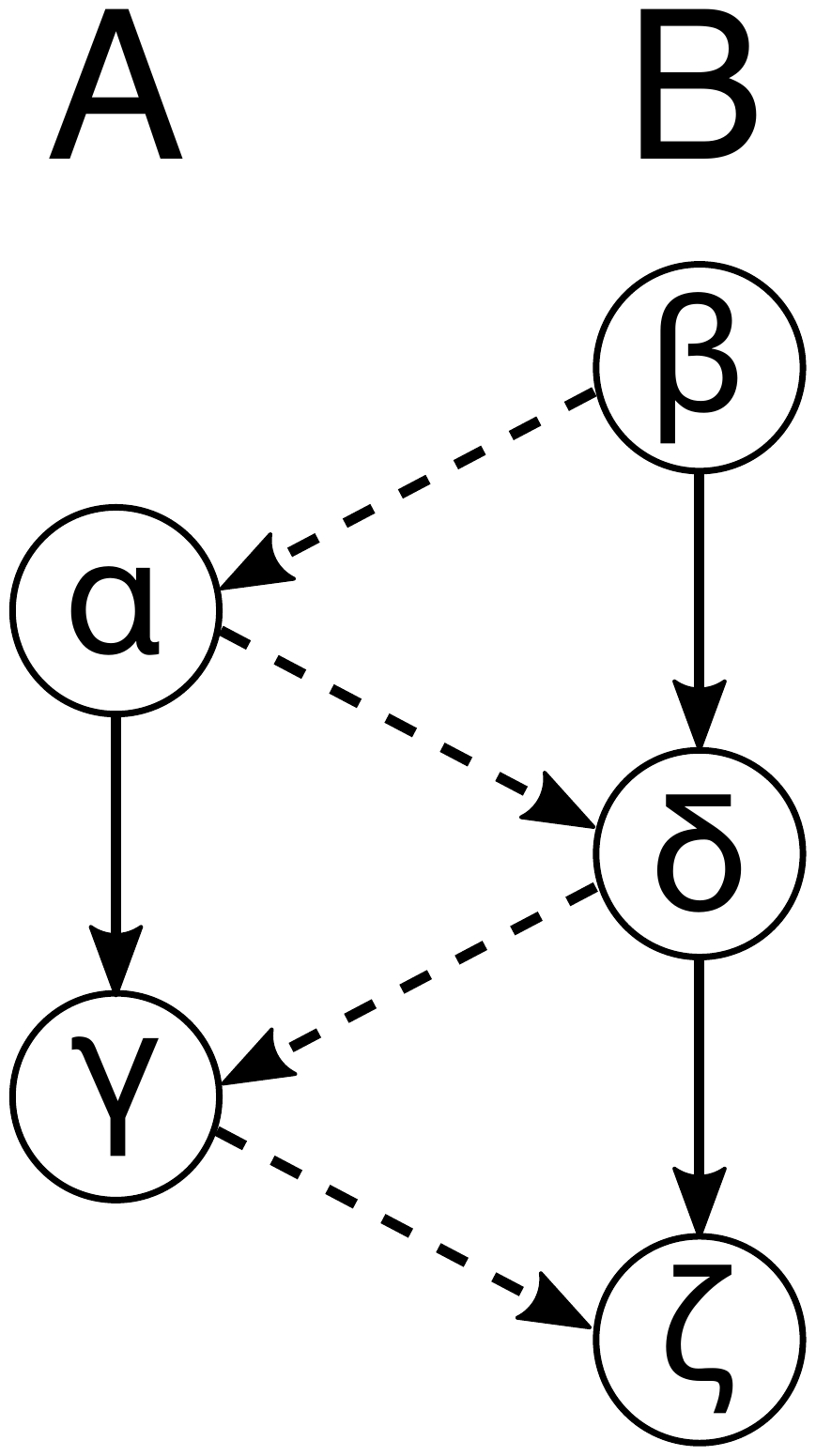} \\
(a) & (b) 
\end{tabular}
\caption{\label{fig:oneway}
Relationships between mutation paths in an MDD: (a) one-way dependent paths $A$ and $B$. Notice that dependency edges only point from one graph (here from $A$ to $B$) to the other and never backwards. (b) interdependent graphs $A$ and $B$. Notice that dependency edges point in both directions. 
}
\end{figure}

\begin{definition} Two paths $P_1$ and $P_2$ that are formed by mutation edges in an MDD $G=(V,E)$ are called {\em one-way dependent} if all edges $(v,w)\in E$ that involve one vertex from from $P_1$ and one from $P_2$ can be oriented in such a way that always the first vertex $v\in P_1$ and the second $w\in P_2$ (or vice versa). Otherwise the paths $P_1$ and $P_2$ are called {\em interdependent}.
\end{definition}

\begin{theorem}
Let $G=(V,E)$ be an MDD that is obtained from a given irreversible program ${\cal P}$. Assume that all mutation paths in $G$ are pairwise one-way dependent. Then the circuit ${\cal C}$ produced by the \textsc{Eager} cleanup method in Algorithm \ref{alg:eager} is correct in the sense that it computes the same function as the original program ${\cal P}$ for some assignment of the input values, ${\cal C}$ is reversible, and all ancillas that may be used by ${\cal C}$ are returned in their initial state at the end of the computation. 
\end{theorem}
\begin{proof} 
It is clear that the circuit computed by ${\cal C}$ is reversible as it uses only reversible gates in its composition. Furthermore, by construction, it computes the given function on some of the output wires. 
It remains to show that all ancilla used during the computation, and inputs which have mutation paths in the MDD G, can be cleaned up, i.e., can be returned to the
initial state. 

Our proof is by induction on the number $k$ of mutation paths $P_1, \ldots, P_k$ inside the MDD $G$. The base case is trivial as if $k=1$ there is only one mutation path and clearly this path either leads to an output, in which case no cleanup is necessary, or leads to a node that has to be cleaned up. However, since all inputs to the path are still available by assumption, the output of the path, i.e., its last node, can simply be moved backward step by step, uncomputing each intermediate result, until the initial state is recovered. 

We strengthen out inductive hypothesis slightly by assuming that we can return the final state of each path $P_1, \ldots, P_k$ to an arbitrary intermediate state. Intuitively, we think of the state ``sliding up and down the mutation path''. Now, we make the inductive step to $k+1$ paths. Recall that a node $v$ can be reversed if all their dependency edges pointing to $v$ are available nodes, i.e., nodes which we can compute from the available data, by sliding the final output either up or down the mutation path. If all dependency arrows of a path point to available nodes, that segment of the path can be reversed. We hence can make the inductive step as follows: Assume that the paths $P_1, \ldots, P_{k+1}$ are arranged in topological order. Let the $v$ be the node holding the current state of the last path $P_{k+1}$, i.e., by assumption on the one-wayness of the graph, all edges point into $P_{k+1}$ and none points backward. Now, consider all nodes in $P_1, \ldots, P_k$ that would have to be available in order to move $v$ one step backward. By induction, we can slide the states for each path into the location that is needed to make this information available for $v$ to move backward one step. Repeating this argument, we see that we can move the state anywhere along the last path $P_{k+1}$ thereby showing that cleanup is possible and algorithm \textsc{Eager} will find this cleanup strategy.
\end{proof}

\subsubsection{Incremental Clean-up}

When the ancilla management realizes that the pool of available ancillas is running low, the compiler looks ahead in the graph to check which ones of the currently allocated bits are needed for the future computation.
We then reverse the circuit creating a checkpoint by copying out the required bits.
If at a later stage the ancilla management again runs out of bits, the process is repeated by taking the last checkpoint to be the beginning of the circuit.
In order to clean up the checkpoints the result of the function may be copied out after which the function can be reversed. Pseudo-code for the incremental clean-up is shown in Figure \ref{alg:increment}.

\begin{algorithm}[hbt]
	\caption{{\textsc{INCREM}} Performs incremental clean-up}
	\begin{algorithmic}[1]
		{\Require A dependency graph $G$ in topological order,
				 $N_q$ number of unused qubits remaining,
		         $n$ index of the current node.
		         $c_p$ previous checkpoint index, $0$ if no previous exists.}
		\State $c_{\text{nodes}} \leftarrow$ $\{v\in$ $G[c_p..n]$ that have edge to a node $>$ $n$)$\}$
		\If{Length($c_{\text{nodes}}$) $\geq$ $N_q$}
			\State $c_{\text{anc}} \leftarrow$ AllocateBits($\text{length}\;c_n$)
			\State $G \leftarrow$ Add Copy from $c_{\text{nodes}}$ to $c_{\text{anc}}$ into $G$ after $n$
			\State $G \leftarrow$ Insert (Reverse $G[c_p..n]$) into $G$ after $(n + \text{length}\;c_{\text{nodes}})$
            \State $c_p \leftarrow n + \text{length}\;c_{\text{nodes}}$
		\EndIf\\
	\Return $G$
	\end{algorithmic}\label{alg:increment}
\end{algorithm}

\subsubsection{Testing the correctness of the compiled Toffoli networks}
We tested all circuits produced by~\REVS{} for correctness and briefly describe our methodology. While general quantum circuits cannot be simulated efficiently on a classical computer, it is possible to simulate reversible circuits that are composed of Toffoli (and CNOT and NOT) gates only. The reason is that all gates that are applied to a given input vector $v_{input} = (x_1, \ldots, x_n)\in \{0,1\}^n$ can be tracked by making local updates, eventually leading to an output vector $v_{output} \in \{0,1\}^n$. In our F$\#$ implementation of~\REVS{} we use this observation to provide the following extremely simple simulator: 

\begin{lstlisting}[language=FSharp]
type Primitive = 
    | RTOFF of int * int * int
    | RCNOT of int * int
    | RNOT of int
    
let simCircuit (gates : Primitive list) (numberOfBits : int) (input : bool list) = 
    let bits = Array.init numberOfBits (fun _ -> false)
    List.iteri (fun i elm -> bits.[i] <- elm) input
    let applyGate gate = 
        match gate with
        | RCNOT(a, b) -> bits.[b] <- bits.[b] <> bits.[a]
        | RTOFF(a, b, c) -> bits.[c] <- bits.[c] <> (bits.[a] && bits.[b])
        | RNOT a -> bits.[a] <- not bits.[a]
    List.iter applyGate gates
    bits
\end{lstlisting}

Our methodology for testing is the following: for a given F$\#$ program ${\cal P}$ which maps a given in put $x \in \{0,1\}^n$ to an output $x \mapsto {\cal P}(x) \in \{0,1\}^m$, we first generate a set $S=\{u_1, \ldots, u_k\} \subseteq \{0,1\}^n$ of test inputs, where typically $k$ is of the order of a few tens or hundred of inputs. We now use the F$\#$ compiler to compile the given program into an executable and run the test inputs in $S$, producing the input/output pairs $(u_i,{\cal P}(u_i))$ for $i=1,\ldots k$. Then we use~\REVS{} to generate the Toffoli network ${\cal T}$ that supposedly implements $(x,y\oplus {\cal P}(x),0)$ and use the above simulator to simulate the output of the Toffoli network on input $(a,b,c) := {\cal T}(u_i,0,0)$. If everything is correct, the outputs coincide, i.e., $a=u_i$, $b={\cal P}(u_i)$, and $c=0$, where $0$ denotes an all-zero vector of the appropriate size.

%
%

\section{Experimental data}\label{sec:data}

\subsection{Arithmetic functions}

A fundamental arithmetic operation that is needed, e.g., as a basic building block to implement the operations required for Shor's algorithm for integer factorization \cite{Shor:97} is integer addition.
There are optimized implementations of adders known for various design criteria, including overall circuit depth \cite{DKR+06} and overall number of ancillas \cite{CDKM:2004}.
Here we compare two families of adders for which hand-optimized reversible implementations are known to the output of the~\REVS{} compiler when applied to a classical (non-reversible) implementation of an adder and while using different cleanup strategies.
Specifically, we choose a space-optimized carry ripple circuit \cite{BCD+:96} that implements $(a,b,0) \mapsto (a,b,a+b)$ whose total number of Toffoli gates for $n$-bit addition modulo $2^n$ scales as $4n-2$ and that requires $3n+1$ qubits.

\begin{figure*}
\begin{lstlisting}[language=FSharp]
let carryRippleAdder (a:bool []) (b:bool []) =
   let n = Array.length a
   let result =  Array.zeroCreate (n)
   result.[0] <- a.[0] <> b.[0]
   let mutable carry = a.[0] && b.[0]
   result.[1] <- a.[1] <> b.[1] <> carry
   for i in 2 .. n - 1 do
      // compute outgoing carry from current bits and incoming carry
      carry <- (a.[i-1] && (carry <> b.[i-1])) <> (carry && b.[i-1])
      result.[i]  <-  a.[i] <> b.[i] <> carry
   result
\end{lstlisting}
\caption{\label{fig:rippleCode} F$\#$ program that implements a carry ripple adder using a simple for loop while maintaining a running carry.}
\end{figure*}

We compare the resulting hand-optimized circuits for implementing carry-ripple adders reversibly to the output generated by the~\REVS{} compiler. To this end, we first implemented a simple carry ripple adder in F$\#$ as shown in Figure \ref{fig:rippleCode}. Note that this is a regular F$\#$ program that can be compiled, e.g., into an executable and run on a conventional computer. Adding reflections allows the~\REVS{} compiler to use the same piece of code and generate an AST, generate the corresponding MDD, apply a given cleanup strategy and emit a corresponding Toffoli network. We applied this for 2 cleanup strategies that are currently implemented, namely the straightforward Bennett strategy that is oblivious to the dependency structure of the program and the eager cleanup strategy that uses the dependency information present in the MDD and tries to cleanup as soon as a variable is no longer needed.

The results of the comparison are summarized in Table \ref{tab:adderSize} and visualized in Figure \ref{fig:ripple}.
The main finding is that the output produced by~\REVS{} is within a constant of the output of the hand-optimized function, both, for the overall circuit size and the overall number of qubits.
Moreover,  it turns out that applying the Bennett strategy leads to sub-optimal scaling in terms of the total number of qubits, whereas the number of gates turns out to be the same for all three kinds of adders, hand-optimized, and the two cleanup strategies investigates.
The classical implementation of the adder is the F$\#$ program shown in Figure \ref{fig:rippleCode} which is then automatically compiled into a Toffoli network using~\REVS{} using a flag for either the Bennett or the eager cleanup strategy.

\begin{figure*}[htb]
  \centering
  \begin{subfigure}[b]{0.45\textwidth}
    \includegraphics[width=\textwidth]{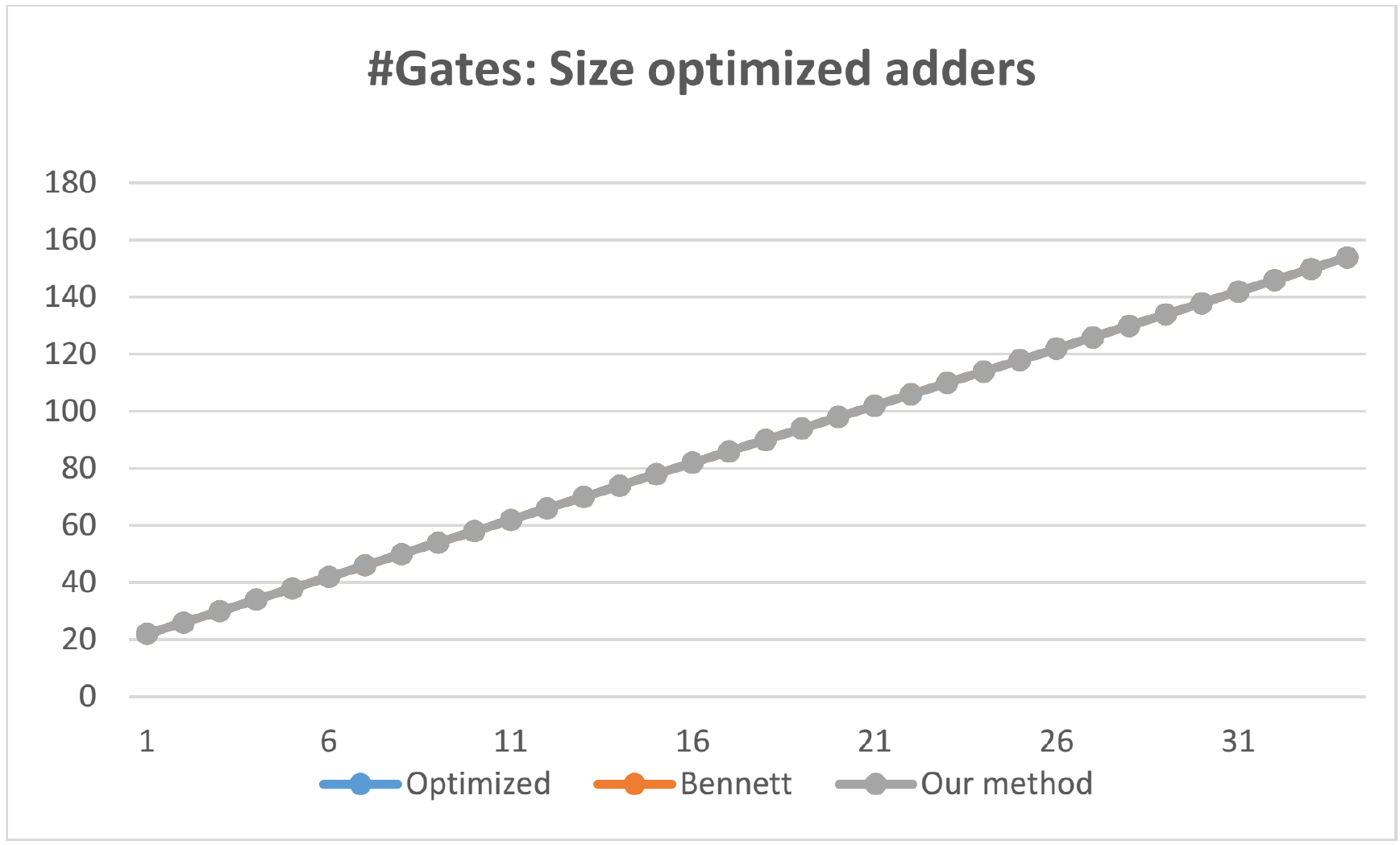}
    \caption{}
  \end{subfigure}
  \hfill
  \begin{subfigure}[b]{0.45\textwidth}
    \includegraphics[width=\textwidth]{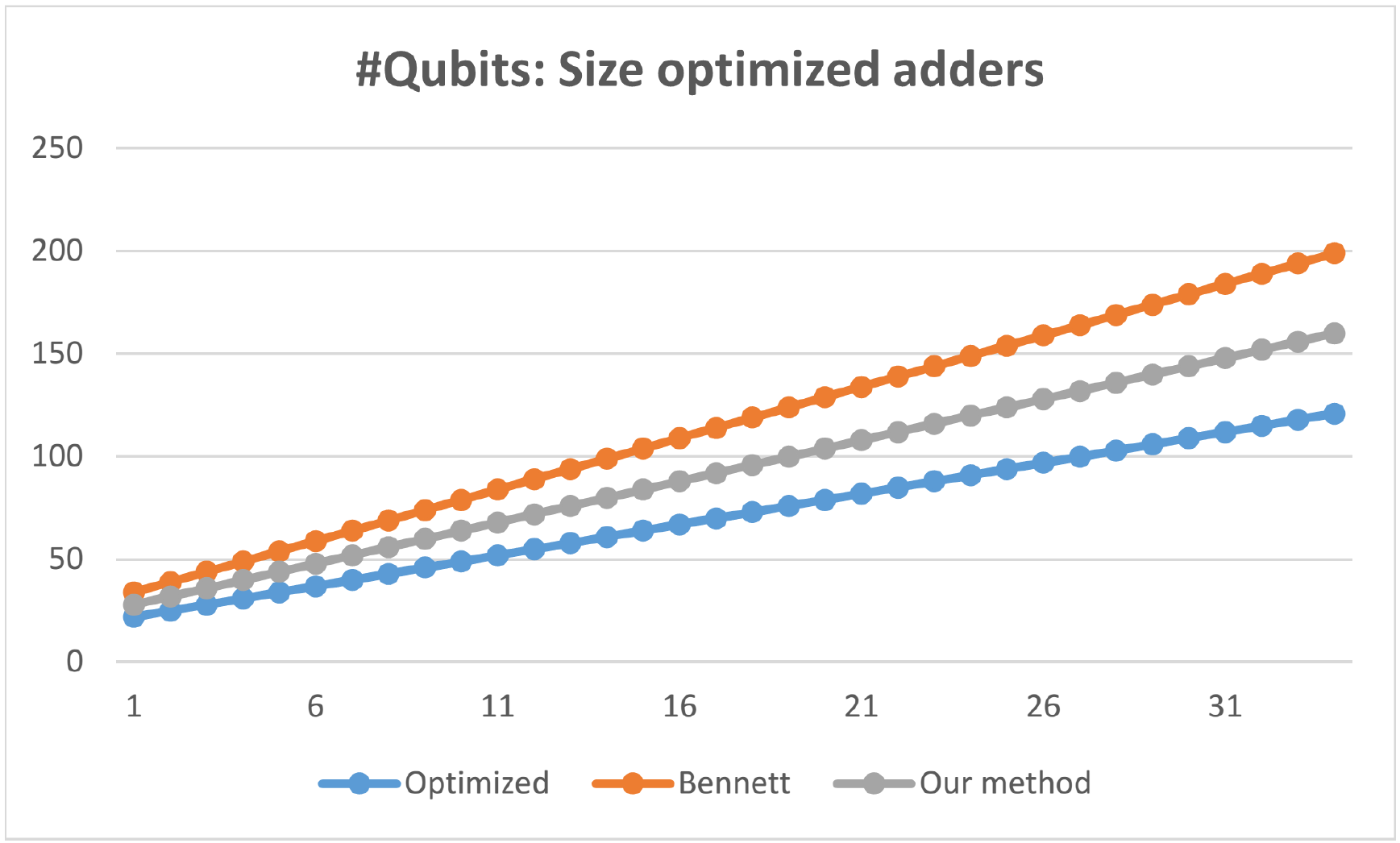}
    \caption{}
  \end{subfigure}
\caption{\label{fig:ripple} Comparison of the resources required to implement carry ripple adders of $n$-bit integers as a reversible circuit. The adders are out-of-place adders, i.e., they map $(a,b,0) \mapsto (a,b, a+b)$, where the addition is performed in the integers modulo $2^n$, i.e., the final carry is ignored. Shown in (a) is the total size of the circuit, as measured by the total number of Toffoli gates needed in the implementation. Shown in (b) is the total number of qubits, including the two input registers, the output register, and the used ancillas.} \end{figure*}

We also choose a depth-optimized adder \cite{DKR+06} that implements integer $n$-bit addition modulo $2^n$ in \[5n-w(n-1)-3\lfloor \lg(n-1) \rfloor -6\] Toffoli gates, where $w(k)$ denotes the Hamming weight of an integer $k$. The number of qubit required for this adder scales as \[4n-w(n-1)-\lfloor \lg(n-1) \rfloor -1.\] The results of the comparison are summarized in Table \ref{tab:adderSize} and visualized in Figure \ref{fig:select}.

\begin{table*}
  \centering
  \begin{tabular}{rrrrrrrrr}
    \toprule
  & \multicolumn{2}{c}{Hand optimized} & \multicolumn{3}{c} {Bennett cleanup}  &  \multicolumn{3}{c} {Eager cleanup} \\
    \cmidrule(r){2-3} \cmidrule(r){4-6} \cmidrule(r){7-9}
    $n$ & $\#$gates & $\#$qubits & $\#$gates & $\#$qubits & time &
    $\#$gates & $\#$qubits & time\\
    \midrule
    10 & 34    & 31       & 34    & 49       & 0.0252 & 34    & 40       & 0.0952  \\
    15 & 54    & 46       & 54    & 74       & 0.0252 & 54    & 60       & 0.0272  \\
    20 & 74    & 61       & 74    & 99       & 0.0342 & 74    & 80       & 0.0382  \\
    25 & 94    & 76       & 94    & 124      & 0.0422 & 94    & 100      & 0.0482  \\
    30 & 114   & 91       & 114   & 149      & 0.0492 & 114   & 120      & 0.0622  \\
    35 & 134   & 106      & 134   & 174      & 0.0592 & 134   & 140      & 0.0722  \\
    40 & 154   & 121      & 154   & 199      & 0.0672 & 154   & 160      & 0.0852  \\
    \bottomrule
  \end{tabular}
  \caption{\label{tab:adderSize} Comparison of different compilation strategies for $n$-bit adders. The optimization criterion is overall circuit size. Shown are the results for a hand-optimized carry ripple adder, an adder that results from applying the~\REVS{} compiler with a cleanup strategy corresponding to Bennett's method, and an adder that results from applying the~\REVS{} compiler with the eager cleanup strategy. Overserve that while the total number of gates is the same for all three methods, the eager cleanup method comes within a space overhead of roughly $33\%$ over the hand optimized adder which is better than the overhead of roughly $66 \%$ for Bennett's method over the hand optimized adder.}
\end{table*}

\begin{figure*}[htb]
\centering
\begin{tabular}{cc}
\!\!\!\includegraphics[width=0.45\columnwidth]{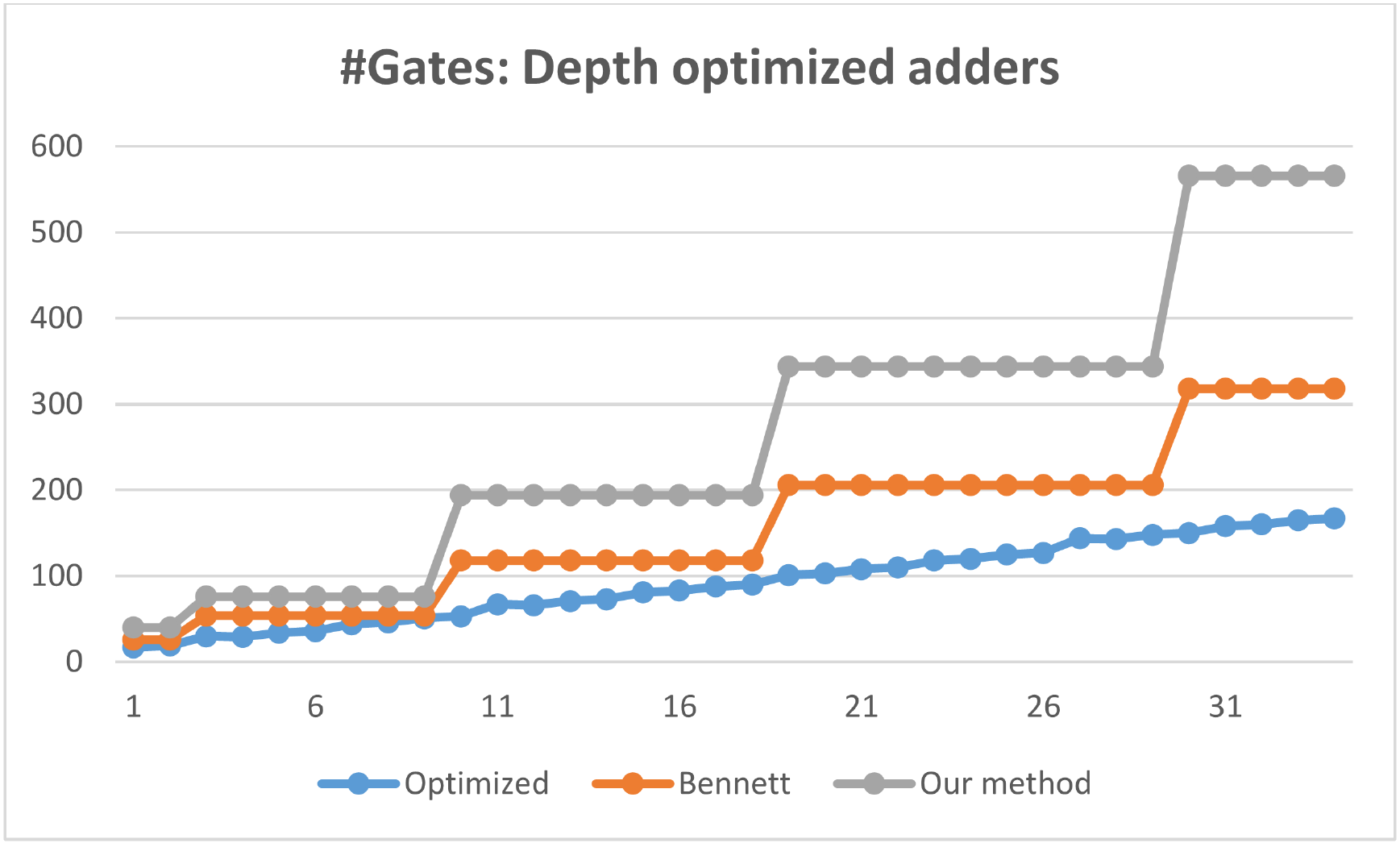} &
\includegraphics[width=0.45\columnwidth]{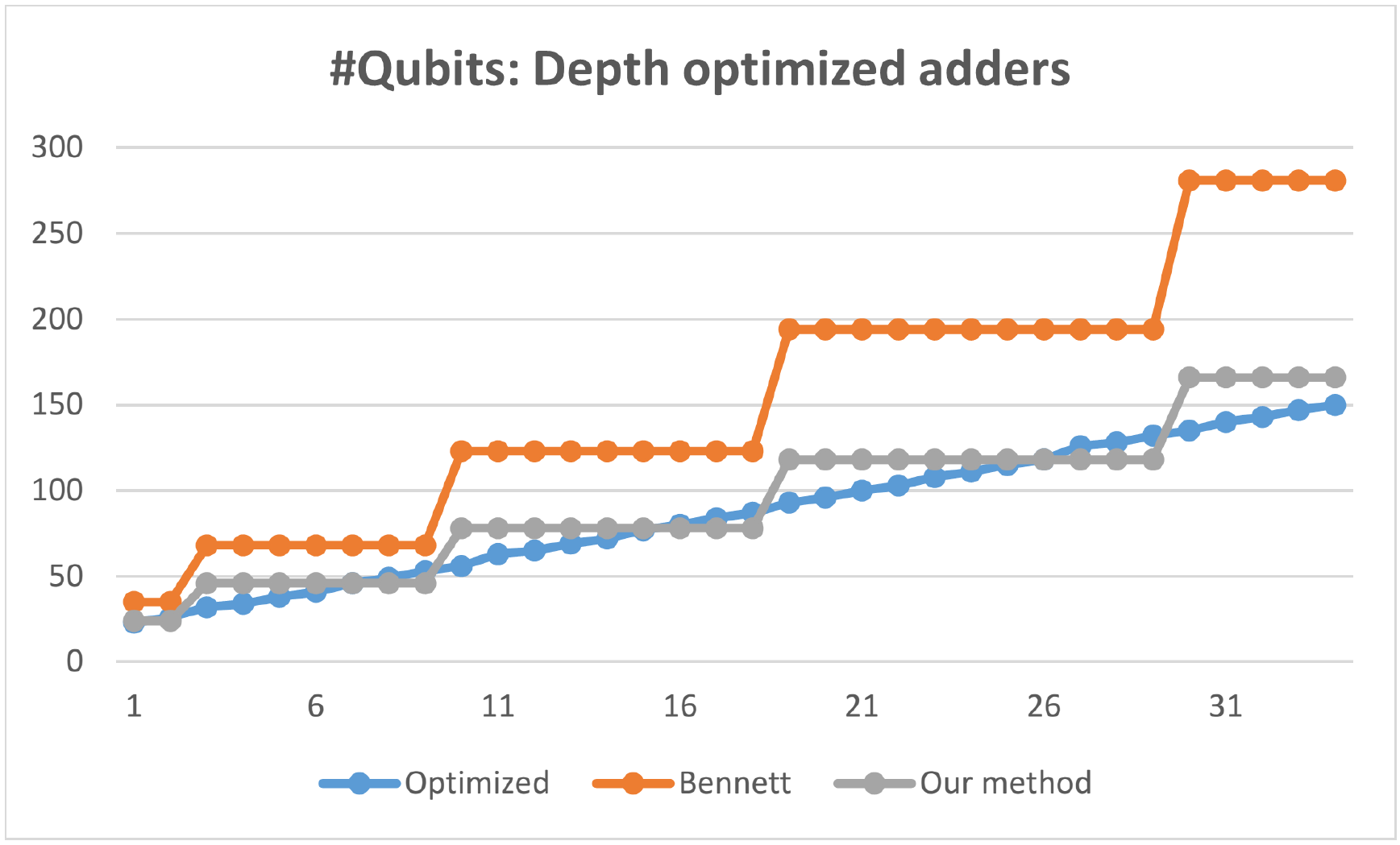}\\
(a) & (b)
\end{tabular}
\caption{\label{fig:select} Comparison of the resources required to implement depth-optimized adders of $n$-bit integers as a reversible circuit. The adders are out-of-place adders, i.e., they map $(a,b,0) \mapsto (a,b, a+b)$, where the addition is performed in the integers modulo $2^n$, i.e., the final carry is ignored. Shown in (a) is the total size of the circuit, as measured by the total number of Toffoli gates needed in the implementation. Shown in (b) is the total number of qubits, including the two input registers, the output register, and the used ancillas.} \end{figure*}

\subsection{Hash functions}

We applied~\REVS{} also to the problem of compiling classical, irreversible implementations of hash functions into Toffoli networks. As an example, we  studied an irreversible implementation of SHA-256 which is a member of the SHA-2 family that hashes a bitstring of arbitrary length to a bitstring of length $256$. Following is a snippet of code from the main loop of the hash function SHA-2. We give more detail about this function in Appendix \ref{app:hash}

\begin{lstlisting}[language=FSharp]
 for i in 0 .. n - 1 do
   // Inplace update of 32 bit registers
   h  <- addMod2_32 (ch e f g)
   h  <- addMod2_32 (s0 a)
   h  <- addMod2_32 w
   h  <- addMod2_32 k
   d  <- addMod2_32 h
   h  <- addMod2_32 (ma a b c)
   h  <- addMod2_32 (s1 e)
   let t = h
   //Reassignment for next iteration
   h <- g; g <- f; f <- e; e <- d
   d <- c; c <- b; b <- a; a <- t
\end{lstlisting}

\begin{figure}[htb]
\centering
\includegraphics[width=0.7\textwidth]{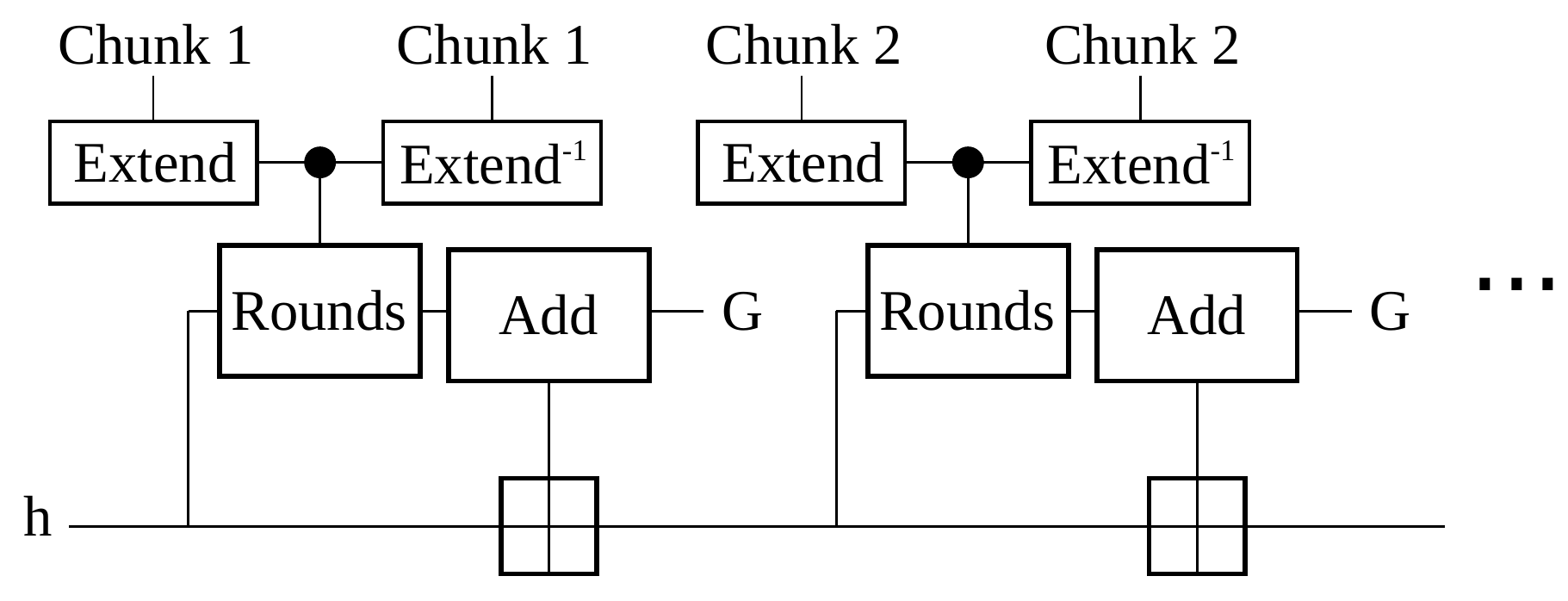}
\caption{\label{fig:sharounds} Data flow diagram corresponding to the SHA-2 cipher.
Note that the cipher has an internal state which gets passed from one round to the next round, leading to garbage qubits, i.e., ancillas that when implemented with a lazy clean-up strategy will accumulate and will lead to a large space overhead.}
\end{figure}

\begin{figure}[htb]
\centering
\includegraphics[width=0.9\columnwidth]{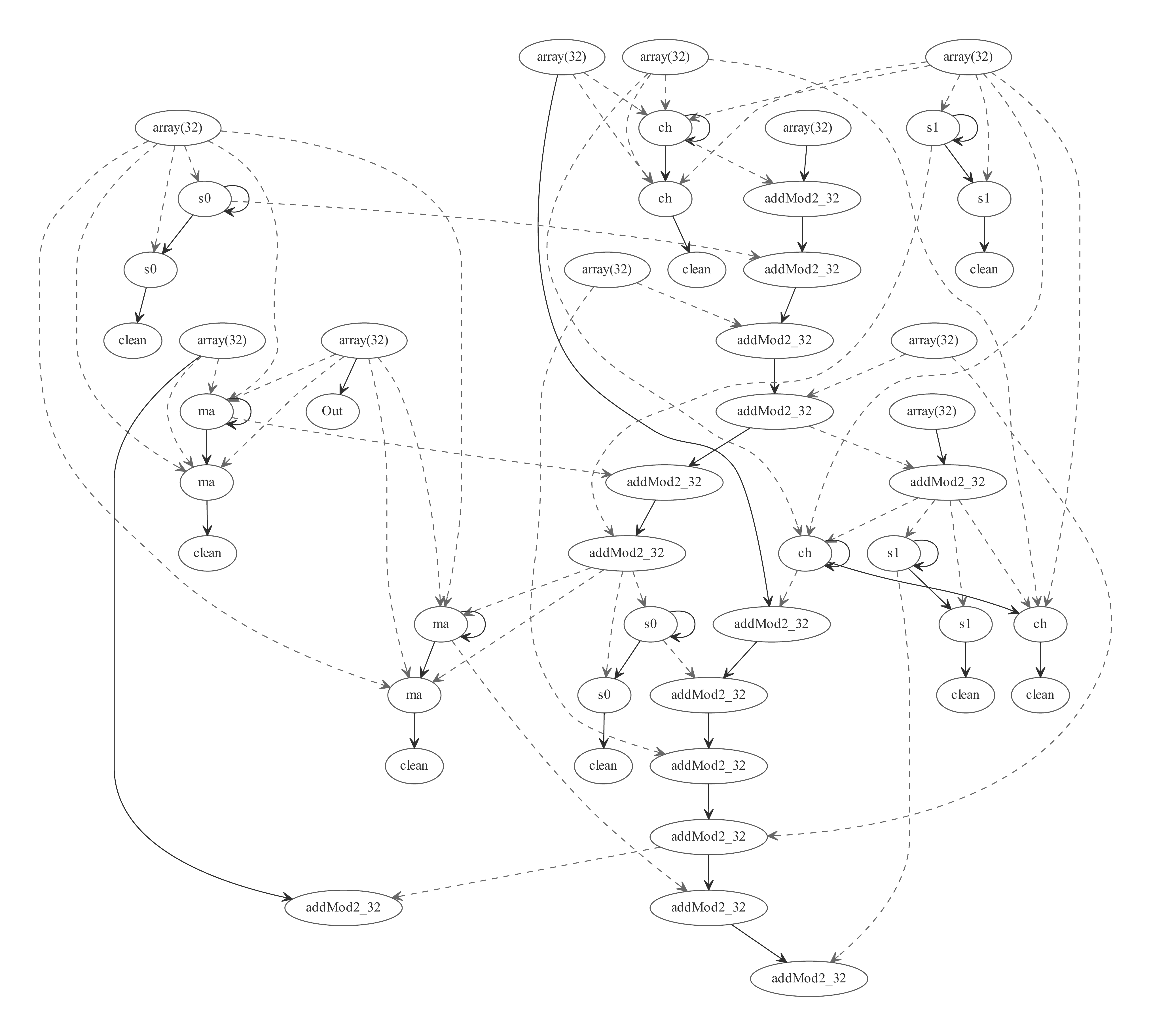}
\caption{\label{fig:mddSha}
MDD for 2 rounds of SHA-2.
}
\end{figure}

In Figure \ref{fig:sharounds} we show how the data flow in the SHA-2 cipher can be visualized. A key feature that allows to keep the amount of space small is that the cipher has an internal state that gets modified within each round by a permutation and then gets passed onto the next round. Based on this we were able to hand-optimize a circuit for SHA-2, see Fig.~\ref{fig:shahand}, however, it turns out that the~REVS{} compiler is able to find a decomposition automatically that comes very close to the hand optimized program. In Table~\ref{tab:SHAcomp} we compare our method with the Bennett method by compiling rounds of the SHA-2 hashing algorithm. 

\begin{table}
  \centering
  \begin{tabular}{rrrrrrrrr}
  \toprule
  & \multicolumn{3}{c} {Bennett cleanup}    &  \multicolumn{3}{c} {Eager cleanup} & \multicolumn{2}{c}{Hand Optimized}  \\
  \cmidrule(r){2-4} \cmidrule(r){5-7} \cmidrule(r){8-9}
  Rounds & $\#$qubits & Toffoli count & time & $\#$qubits & Toffoli count & time & $\#$qubits & Toffoli count  \\
  \midrule
  1  & 704   & 1124     & 0.2546002 & 353   & 690      & 0.3290822 & 353   & 683      \\
  2  & 832   & 2248     & 0.2639522 & 353   & 1380     & 0.3360352 & 353   & 1366     \\
  3  & 960   & 3372     & 0.2823012 & 353   & 2070     & 0.3420732 & 353   & 2049     \\
  4  & 1088  & 4496     & 0.2827132 & 353   & 2760     & 0.3543582 & 353   & 2732     \\
  5  & 1216  & 5620     & 0.2907102 & 353   & 3450     & 0.3664272 & 353   & 3415     \\
  6  & 1344  & 6744     & 0.3042492 & 353   & 4140     & 0.3784522 & 353   & 4098     \\
  7  & 1472  & 7868     & 0.3123962 & 353   & 4830     & 0.3918812 & 353   & 4781     \\
  8  & 1600  & 8992     & 0.3284542 & 353   & 5520     & 0.4025412 & 353   & 5464     \\
  9  & 1728  & 10116    & 0.3341342 & 353   & 6210     & 0.4130702 & 353   & 6147     \\
  10 & 1856  & 11240    & 0.3449002 & 353   & 6900     & 0.4304762 & 353   & 6830     \\
  \bottomrule
  \end{tabular}
  \caption{\label{tab:SHAcomp} Comparison of different compilation strategies for the cryptographic hash function SHA-2.
    Shown are the resulting circuit size, measured by the total number of Toffoli gates, the resulting total number of qubits, and the time it took to compile the circuit for various numbers of rounds.
    All timing data are measure in seconds and resulted from running the F$\#$ compiler in Visual Studio 2013 on an Intel i7-3667 @ 2GHz 8GB RAM under Windows 8.1. The table shows significant savings of almost a factor of $4$ in terms of the total numbers of qubits required to synthesize the cipher when comparing the simple Bennett cleanup strategy versus the eager cleanup strategy.
    The reason for this is that the Bennett cleanup methods allocates new space essentially for each gates versus the eager cleanup strategy that tries to clean up and reallocate space as soon as possible which for the round-based nature of the function can be done as soon as the round is completed.}
\end{table}

\begin{figure}[htb]
    \centering
    \includegraphics[width=\columnwidth]{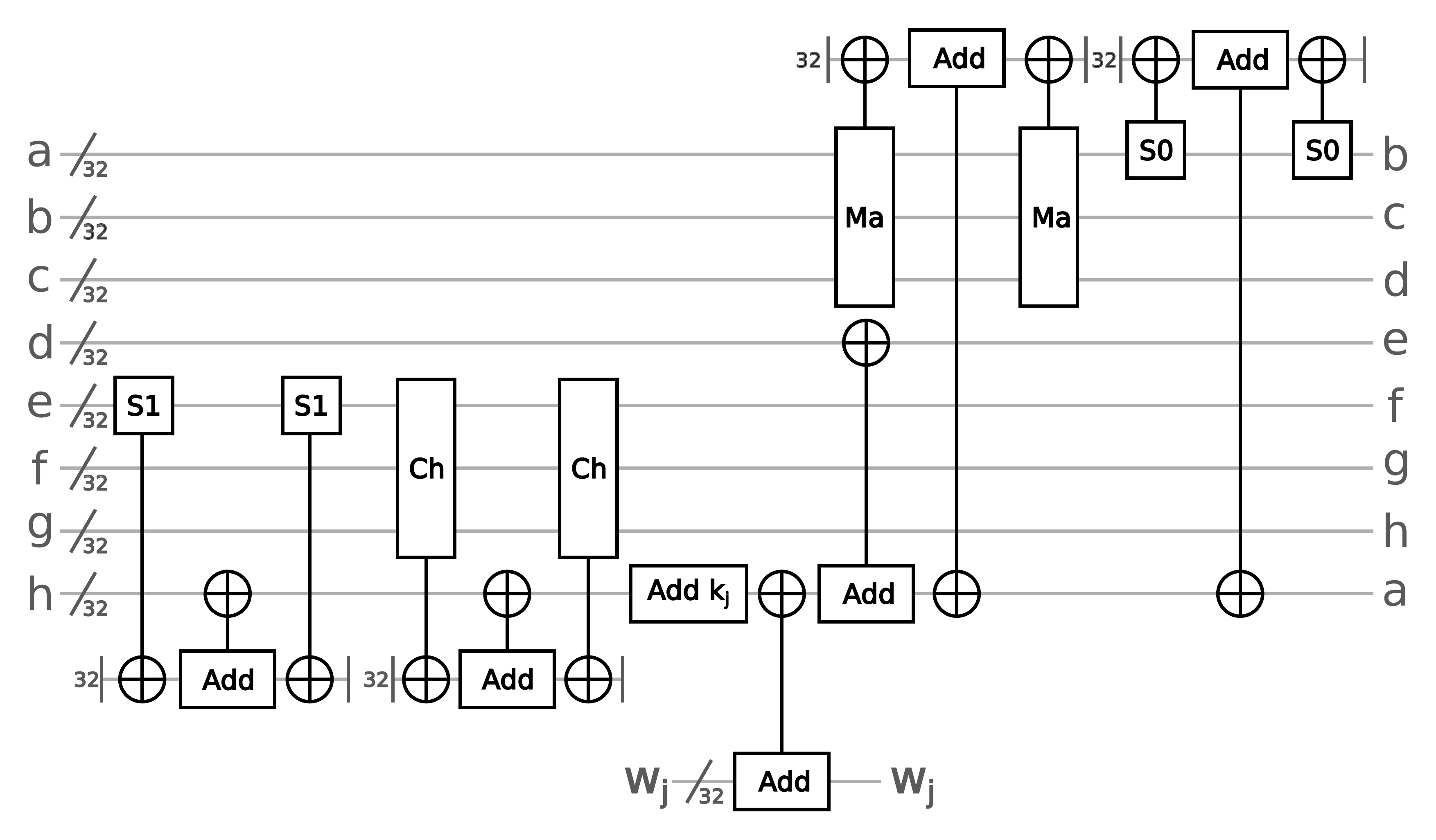}
    \caption{\label{fig:shahand}
        Description of a quantum circuit for SHA-2 that was obtained by inspection of the cipher and translation into a structured circuit.
        This circuit contains 7 adders (mod $2^{32}$).
        Using the adder from CDKM\cite{CDKM:2004} with a Toffoli cost of $2n-3$ this corresponds to 61 Toffoli gates per adder or 427 per round.
        Note that the $K_j$ adder is a constant adder so may use slightly less gates if we remove controls where the value is known.
        Each MA operation can be implemented with 2 Toffoli gates per bit for a total of 128 Toffoli gates for both operations above.
        Each CH operation can also be implemented for a Toffoli cost of 2 leading to a further 128 Toffoli gates per round.
        The S0 and S1 operations can be implemented using only Toffoli gates.
        The total per round Toffoli cost is therefore $427+128+128=683$
}
\end{figure}

\subsection{BLIF benchmarks}

The Berkeley Logic Interchange Format (BLIF) \cite{blif} is a logic level circuit description of a classical operation. The format allows the specification of hierarchical logical circuits, based on a simple text input form. Circuits can have combinational components, which typically are given by a collection of truth tables using separate lines for each input/output combinations, where ``don't cares'' are allows, and sequential components, i.e., memory elements such as latches. A BLIF file has the following structure:

\begin{lstlisting}[language=Fsharp]
.model decl-model-name
.inputs decl-input-list 
.outputs decl-output-list 
.clock decl-clock-list
<command> 
... 
<command> 
.end
\end{lstlisting}

Here .model is a string describing the name of the circuit, .inputs is a list of input bits, .outputs a list of output bits, .clock is a list of clocks to describe timing information, and each command is a specification of part of the logical description of the gate. As we are interested in combinational circuits only, i.e., circuits without memory, for us only a subset of all expressible BLIF syntax is relevant. 

A typical command statement them looks like a collection of lines, each having the form as follows: 

\begin{lstlisting}
-01-
\end{lstlisting}

Where a zero or one in the nth column means that the corresponding input must be zero or one and a dash means that any input is accepted. In the above example, the specification implies that any input in the set $\{ (0,0,1,0),(0,0,1,1),(1,0,1,0),(1,0,1,1)\}$ evaluates to $1$ and all other inputs evaluate to $0$. 

If any one statement in a block is true then the entire block is true.
For example given the inputs $x_1$, $x_2$, and $x_3$ the following block has the value
$(\lnot x_1\land x_3)\lor(\lnot x_2)\lor(x_1\land x_2 \land x_3)$:

\begin{lstlisting}
0-1
-0-
111
\end{lstlisting}

The BLIF format is attractive as there is a rich set of circuit benchmarks that have been used primarily by the Circuit and Systems community in the 80s and 90s. These benchmarks include the so-called ISCAS'85, ISCAS'89,  MCNC'91, LGSynth'91 and LGSynth'92 collections \cite{benches,MCNC91,LGSynth91}. We identified all examples from the union of these benchmarks that {\em only} use combinational circuit elements, i.e., for all those Boolean functions in principle a reversible circuit can be computed. We obtained a set of $135$ benchmark circuits which we used to test the performance of our Boolean expression generation subroutines in~\REVS. 

\subsubsection{Optimization of BLIF generation: Converting XOR to OR}
In the case of mutually exclusive statements XOR is equivalent to OR.
That is to say $a \lor b = a \oplus b$ if $a = 1 \implies b = 0$ and $b = 1 \implies a =0$.
For example $a\land b \lor \neg a \land c = a\land b \oplus \neg a \land c$.

This is useful as it allows us to avoid the use of Toffoli gates and use less ancillas. For example if we wish to compute $a \lor b \lor c$ we might use the circuit:

  \[
    \Qcircuit @C=1em @R=.7em {
        \lstick{a} & \ctrlo{1} & \qw      & \ctrlo{1} & \rstick{a}\qw\\
        \lstick{b} & \ctrlo{2} & \qw      & \ctrlo{2} & \rstick{b}\qw\\
        \lstick{c} & \qw       & \ctrlo{2}& \qw       & \rstick{c}\qw\\
        \lstick{0} & \targ     & \ctrl{1} & \targ     & \rstick{0}\qw \\
        \lstick{0} & \qw       & \targ    & \targ      & \rstick{a \lor b \lor c}\qw
    }
  \]

\medskip
And as another example the function $a \oplus b \oplus c$ can be computed as:

\[
    \Qcircuit @C=1em @R=.7em {
        \lstick{a} & \ctrl{3}  & \qw      & \qw      & \rstick{a}\qw\\
        \lstick{b} & \qw       & \ctrl{2} & \qw      & \rstick{b}\qw\\
        \lstick{c} & \qw       & \qw      & \ctrl{1} & \rstick{c}\qw\\
        \lstick{0} & \targ     & \targ    & \targ    & \rstick{a \oplus b \oplus c}\qw \\
    }
\]
\medskip

In general, given a set of AND expressions that are combined using OR we want to find sets of mutually exclusive statements that minimize the use of AND. We consider each AND expression to be a vertex on a graph and add edges between vertices that are mutually exclusive. Now we cover this graph using the smallest possible number of 
cliques using an algorithm that solves the so-called CLIQUE-COVER problem which asks to partition the vertices of a graph can be partitioned into cliques. The NP-completeness of clique cover for a given upper bound $k$ of allowed cliques is well-known \cite{GJ:79}, however, good practical approximation algorithms exist \cite{GHH+:96}. 

After finding a cliques partition each set of mutually exclusive statements can be implemented by evaluating the AND statements and combining all of the values on a single ancilla using XOR for each clique. These results can then be combined using OR statements. We can preprocess the given BLIF files in such a way that the cliques will be grouped in the output. This yields a new BLIF file, however, the effect of the reordering is that instead of OR functions now the much cheaper XOR functions can be used. Using the two earlier examples given above, we can illustrate this technique by applying it to the following circuit: 

\begin{center}
  \begin{tabular}{c c}
    \begin{lstlisting}
0-1
-0-
111
    \end{lstlisting} &
    \begin{lstlisting}
-0-
0-1
111
    \end{lstlisting} \\
    Before Reordering & After Reordering
  \end{tabular}
\end{center}

After reordering, we can now replace the OR of the rows by the XOR of the row. We obtain the following circuit from compiling the reordered BLIF circuit using~\REVS:

\begin{center}
\includegraphics[width=0.9\columnwidth]{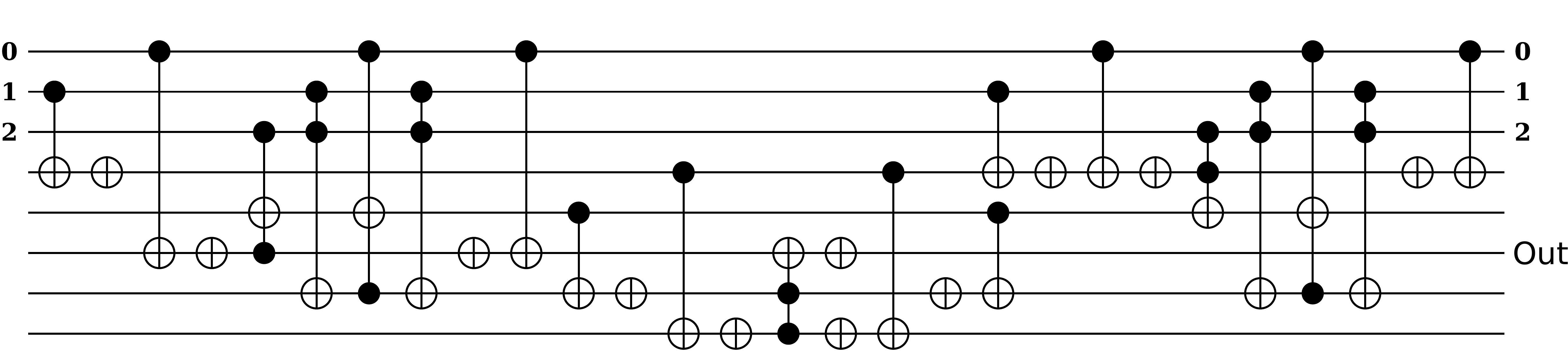}
\end{center}

We implemented the procedure that first performs the offline conversion of the given BLIF circuit to an equivalent BLIF circuit by performing the clique-cover-based XOR maximization. Then this circuit is converted directly into an MDD before cleanup and in doing so, the~\REVS{} compiler finds the optimized grouping that replaces OR terms with XOR terms. We applied this method to the above mentioned suite of benchmarks. 
Shown in Figure \ref{blifChart} the results of applying eager cleanup vs Bennett cleanup to the resulting MDD for $38$ benchmarks out of our overall suite of $135$ benchmarks. 

\begin{table}
  \centering
  \begin{tabular}{rrrrrcr}
    \toprule
    & \multicolumn{2}{c} {Eager cleanup}    &  \multicolumn{2}{c} {Bennett cleanup} \\
    \cmidrule(r){2-3} \cmidrule(r){4-5}
    Name            & Bits          &   Gates       &      Bits     &       Gates  & Qubit reduction ($\%$) & Time   \\
    \midrule
  ADDERFDS        &       76      &       3394    &       149     &       2485    &   48.99       & 0.0222 \\    
    CM150           &       42      &       1696    &       71      &       1017    &  40.85      & 0.0092 \\
    CM151           &       25      &       795     &       36      &       485     &  30.56      & 0.0032 \\
    CM163           &       36      &       843     &       66      &       825     & 45.45       & 0.0072 \\
    DES             &       992     &       85351   &       2123    &       62425   &  53.27      & 1.7202 \\
    PARITYFDS       &       33      &       290     &       33      &       301     & 0.00           & 0.0052 \\
    alu4cl         &       175     &       13487   &       441     &       8867    &  60.32       & 0.1012 \\
    apex7           &       111     &       3299    &       242     &       3077    &   54.13     & 0.0312 \\
    b9              &       127     &       4444    &       316     &       3185    &  59.81      & 0.0352 \\
    c8              &       80      &       4791    &       213     &       4149    &   62.44     & 0.0222 \\
    cht             &       89      &       2664    &       179     &       3025    &  50.28      & 0.0192 \\
    cmb             &       37      &       1042    &       60      &       707     &  38.33      & 0.0082 \\
    comp            &       97      &       2606    &       162     &       1991    &  40.12      & 0.0172 \\
    cordic          &       130     &       3094    &       262     &       2363    &  50.38      & 0.0372 \\
    count           &       71      &       1304    &       124     &       1435    &  42.74      & 0.0142 \\
    cu              &       36      &       1395    &       80      &       1045    &  55.00         & 0.0092 \\
    dalu            &       1056    &       77482   &       3444    &       49463   &  69.34      & 1.4402 \\
    decod           &       27      &       188     &       28      &       365     &  3.57       & 0.0032 \\
    f51m            &       49      &       4684    &       69      &       2427    & 28.99       & 0.0592 \\
    frg2            &       531     &       44820   &       2080    &       34117   & 74.47       & 0.4722 \\
    i10             &       1633    &       24505   &       3385    &       23489   & 51.76       & 7.5782 \\
    lal             &       86      &       4051    &       212     &       2939    &  59.43      & 0.0262 \\
    mux             &       45      &       2841    &       67      &       1473    &   32.84     & 0.0122 \\
    pair            &       666     &       37630   &       2137    &       27091   & 68.83       & 0.8632 \\
    pcler8cl       &       57      &       507     &       78      &       599      & 26.92       & 0.0062 \\
    pm1             &       41      &       1248    &       96      &       987     &  57.29      & 0.0102 \\
    rot             &       325     &       22407   &       779     &       12709   & 58.28       & 0.1732 \\
    sct             &       64      &       4509    &       202     &       2923    & 68.32       & 0.0242 \\
    t481            &       2096    &       134542  &       6774    &       95393   & 69.06       & 3.4702 \\
    tcon            &       37      &       200     &       51      &       289     & 27.45       & 0.0032 \\
    term1           &       181     &       20744   &       724     &       11881   & 75.00          & 0.1072 \\
    toolarge       &       707     &       155854  &       1191    &       155885   &  40.64      & 3.7072 \\
    top             &       2676    &       33190   &       3940    &       30981   &  32.08      & 4.7372 \\
    ttt2            &       116     &       10155   &       400     &       8511    & 71.00          & 0.0622 \\
    unreg           &       58      &       1264    &       103     &       1129    &  43.69      & 0.0112 \\
    vda             &       173     &       19234   &       899     &       19705   & 80.76       & 0.1402 \\
    x3              &       377     &       30258   &       1309    &       21467   & 71.2        & 0.2302 \\
    z4ml            &       45      &       3628    &       60      &       1887    & 25.00          & 0.0242 \\   \bottomrule
  \end{tabular}
  \caption{A selection of combinational circuits from a collection of BLIF benchmarks from the classical Circuits and Systeme community, including ISCAS'85, ISCAS'89, MCNC'91, LGSynth'91, LGSynth'92. Shown are the qubit and gate costs for the eager cleanup method and in comparison the corresponding cost for the Bennett cleanup. Typically, we obtain a qubit reduction of around $50\%$ or more at the price of only a moderate increase of the total number of Toffoli gates that are needed to implement the reversible circuits. Also shown is the compilation time with respect to an Intel i7-3667 @ 2GHz 8GB RAM processors under Windows 8.1.
  \label{blifChart}
  }
\end{table}

\section{Discussion}\label{discussion}

\subsection{Comparison with Bennett at function boundaries}

A simple alternative cleanup method is to just to apply the original Bennett compute-copy-uncompute method at function boundaries, i.e., whenever a result of a computation is no longer needed, the computation is immediately cleaned up. This method will preform reasonably well in cases where no in-place functions are used. Our primitive operation (the boolean expression) allows for in-place operation though. Bennett-style cleanup  assumes that all operations are out of place and misses out on many opportunities to save on both space and time.

In contrast, consider the case with in-place functions: suppose we have an input ($a$,$b$) and we preform some calculation arriving at $f(b)$.
Now we use an in-place function which maps $(f(b),a) \mapsto (f(b),g(f(b),a))$.
$g(f(b),a)$ is then set as the output of the function.
Using our method we need only clean up $f(b)$ and to create a new in-place function $(a,b) \mapsto (g(f(b),a),b)$.

Even if we are not seeking to create a new in-place function this cleanup strategy can be useful.
Consider the case where the function input is $(a,b)$ and the output is $(a,b,c)$ (with $c$ being allocated inside the function).
Lets say we first apply some function out of place to $a$, $(a,0) \mapsto (a,f(a))$, and preform some other calculation arriving at $f(b)$.
Then we want to use an in-place operation to map $(g(b),f(a)) \mapsto (g(b),h(g(b),f(a)))$.
Using our method we would again only need to clean $f(b)$ to produce the function $(a,b) \mapsto (a,b,c)$ where $c= h(g(b),f(a))$.

A good example for benefits to using in-place updates is the mentioned examples of the SHA-2 algorithm (see Figure~\ref{fig:shahand}).
In the implementation used below functions are calculated at each iteration only to be added in-place to the result.
Using our method we can immediately cleanup those functions since they are not needed after the addition is preformed.
This prevents additional ancilla from being used at each iteration.
Even if each iteration were wrapped in a function the total number of bits used would be higher.
Looking again at Figure \ref{fig:shahand} we see that in the hand implementation each adder can be cleaned up before bits are allocated for the next.
Using the Bennett method with the function boundary at each iteration all of the adders would be cleaned up at the end so the total number of bit needed would be greater.

The graph structure also provides information which could be used in other possible improvements to our method.
For example if we wanted trade of some time for space we could temporally clean up some bits and recreate them later.
When doing this we would want to choose bits that are both easy to compute and have a large gap until the next time they are used in the computation.
The graph structure allows us to quantify both of these metrics.
It also allows us to easily generate the cleanup and re-computation strategies.

Finally, we point out that while in some cases cleaning up at the function boundaries can lead to useful resource reductions, in other cases it can lead to unacceptable time-overheads. It is known for instance in cases of some recursions such as Karatsuba's algorithm for multiplication of $n$-bit integers that by performing a cleanup at the function boundaries, the advantage of a subquadratic algorithm is lost \cite{KPF:2006,SS:2015}, whereas by performing cleanup at the very end of the recursion one can still achieve an algorithm that asymptotically scales in time as $\Theta(n^{\log_2(3)})$. 

\subsection{Conclusions} We presented~\REVS, a compiler and programming language that allows to automate the translation of classical, irreversible programs into reversible programs. Contrary to previous approaches of  reversible programming languages such as the reversible languages R or Janus \cite{Perumalla:2014}, our language does not constrain the programmer. Also, in contrast to previous approaches for implementing Bennett-style strategies such as Quipper \cite{GLR+:2013a,GLR+:2013b} our approach is significantly more space efficient, in some practically relevant cases (hash functions) the savings of our method compared to Bennett-style approaches can even be unbounded. We navigate around the PSPACE-completeness of finding the optimal pebble game by invoking heuristic strategies which seek to identify parts of the program that lead to mutation which then can be implemented via in-place operations. In order to manage the arising data dependencies, we introduced MDD graphs which capture data dependencies as well as data mutation in one graph. We prove that our eager cleanup strategy is correct, provided the mutation paths that occur in the MDD have no inter-path dependency. In case such dependencies arise, we clean up the paths using the standard Bennett strategy, which allows us to compile any program that can be expressed in~\REVS{} into a Toffoli network. 

We found examples where our dependency-graph based method for eager cleanup is better than Bennett's original method, even when Bennett's method is implemented by cleaning up at function boundaries. 

Using an example benchmark suite compiled from classical circuits and systems community, we show that the method can be applied for medium to large scale problems. We also show that hash functions such as SHA-2 and MD5 can be compiled into reversible circuits with ease. 



\begin{thebibliography}{10}

\bibitem{Bennett:73}
C.~H. Bennett,
\newblock ``Logical reversibility of computation,''
\newblock {\em IBM Journal of Research and Development}, vol. 17, pp. 525--532,
  1973.

\bibitem{Markov:2014}
I.~L. Markov,
\newblock ``Limits on fundamental limits to computation,''
\newblock {\em Nature}, vol. 512, pp. 147--154, 2014.

\bibitem{Shor:97}
P.~W. Shor,
\newblock ``Polynomial-time algorithms for prime factorization and discrete
  logarithms on a quantum computer,''
\newblock {\em SIAM Journal on Computing}, vol. 26, no. 5, pp. 1484--1509,
  1997.

\bibitem{Grover:96}
L.~Grover,
\newblock ``{A fast quantum mechanical algorithm for database search},''
\newblock in {\em Proceedings of the Symposium on Theory of Computing
  (STOC'96)}. 1996, pp. 212--219, ACM Press.

\bibitem{BCC+:2014}
Berry~D. W., A.~M. Childs, R.~Cleve, R.~Kothari, and R.~D. Somma,
\newblock ``{Exponential improvement in precision for simulating sparse
  Hamiltonians},''
\newblock in {\em Symposium on Theory of Computing, (STOC'14)}, 2014, pp.
  283--292.

\bibitem{WS:2014}
D.~Wecker and K.~Svore,
\newblock ``{LIQ{\em ui}$|\rangle$: a software design architecture and
  domain-specific language for quantum computing},''
\newblock arXiv.org preprint arXiv:1402.4467.

\bibitem{GLR+:2013a}
A.~S. Green, P.~{LeFanu Lumsdaine}, N.~J. Ross, P.~Selinger, and B.~Valiron,
\newblock ``{An introduction to quantum programming in Quipper},''
\newblock in {\em Proc. Reversible Computation (RC'13)}. 2013, ACM.

\bibitem{GLR+:2013b}
A.~S. Green, P.~{LeFanu Lumsdaine}, N.~J. Ross, P.~Selinger, and B.~Valiron,
\newblock ``{Quipper: a scalable quantum programming language},''
\newblock in {\em Proc. Conference on Programming Language Design and
  Implementation (PLDI'13)}. 2013, ACM.

\bibitem{Omer:2000}
B.~{\"O}mer,
\newblock {\em {Quantum programming in QCL}},
\newblock Master's thesis, Technical University of Vienna, 2000.

\bibitem{AG:2009}
Th. Altenkirch and A.~S. Green,
\newblock ``{The quantum IO monad},''
\newblock in {\em {Semantic Techniques in Quantum Computation}}, {S. Gay and I.
  Mackie}, Ed., pp. 173--205. Cambridge University Press, 2009.

\bibitem{SV:2009}
P.~Selinger and B.~Valiron,
\newblock ``{Quantum lambda calculus},''
\newblock in {\em {Semantic Techniques in Quantum Computation}}, {S. Gay and I.
  Mackie}, Ed., pp. 135--172. Cambridge University Press, 2009.

\bibitem{LST+:2013}
A.~Lapets, M.~{da Silva}, M.~Thome, A.~Adler, J.~Beal, and M.~Roetteler,
\newblock ``{QuaFL: A typed DSL for quantum programming},''
\newblock in {\em {Proceedings of Workshop on Functional Programming Concepts
  in Domain-Specific Languages (FPCDSL'13)}}. 2013, pp. 19--26, ACM Press.

\bibitem{LR:2013}
A.~Lapets and M.~Roetteler,
\newblock ``{Abstract resource cost derivations for logical quantum circuit
  descriptions},''
\newblock in {\em {Proceedings of Workshop on Functional Programming Concepts
  in Domain-Specific Languages (FPCDSL'13)}}. 2013, pp. 35--42, ACM Press.

\bibitem{HPJ+:2015}
J.~Heckey, S.~Patil, A.~{Javadi Abhari}, A.~Holmes, D.~Kudrow, K.~R. Brown,
  D.~Franklin, F.~T. Chong, and M.~Martonosi,
\newblock ``Compiler management of communication and parallelism for quantum
  computation,''
\newblock in {\em Proceedings of the Twentieth International Conference on
  Architectural Support for Programming Languages and Operating Systems
  (ASPLOS'15)}. 2015, pp. 445--456, ACM.

\bibitem{JFJ+:2012}
A.~{Javadi Abhari}, A.~Faruque, M.~{Javad Dousti}, L.~Svec, O.~Catu,
  A.~Chakrabati, Ch.-F. Chiang, S.~Vanderwilt, J.~Black, F.~T. Chong,
  M.~Martonosi, M.~Suchara, K.~Brown, M.~Pedram, and T.~Brun,
\newblock ``{Scaffold: quantum programming language},''
\newblock {Princeton University Technical Report}, 2012.

\bibitem{YG:2007}
T.~Yokoyama and R.~Gl{\"u}ck,
\newblock ``A reversible programming language and its invertible
  self-interpreter,''
\newblock in {\em Proc. Workshop on Partial Evaluation and Semantics-based
  Program Manipulation (PEPM'07)}. 2007, pp. 144--153, ACM.

\bibitem{Thomsen:2012}
M.~K. Thomsen,
\newblock ``A functional language for describing reversible logic,''
\newblock in {\em Proc. Forum on Specification and Design Languages (FDL'12)}.
  2012, pp. 135--142, IEEE.

\bibitem{Perumalla:2014}
K.~S. Perumalla,
\newblock {\em Introduction to Reversible Computing},
\newblock CRC Press, 2014.

\bibitem{YAG:2008a}
T.~Yokoyama, H.~B. Axelsen, and R.~Gl\"{u}ck,
\newblock ``Principles of a reversible programming language,''
\newblock in {\em Proceedings of the 5th conference on Computing frontiers
  (CF'08)}. 2008, pp. 43--54, ACm Press.

\bibitem{YAG:2008b}
T.~Yokoyama, H.~B. Axelsen, and R.~Gl{\"{u}}ck,
\newblock ``Reversible flowchart languages and the structured reversible
  program theorem,''
\newblock in {\em Proceedings 35th International Colloquium on Automata,
  Languages and Programming (ICALP'08)}, 2008, pp. 258--270.

\bibitem{MMD:2003}
D.~M. Miller, D.~Maslov, and G.~W. Dueck,
\newblock ``A transformation based algorithm for reversible logic synthesis,''
\newblock in {\em Proceedings of the 40th Design Automation Conference
  (DAC'03)}, 2003, pp. 318--323.

\bibitem{MMD:2007}
D.~Maslov, D.~M. Miller, and G.~W. Dueck,
\newblock ``{Techniques for the synthesis of reversible Toffoli networks},''
\newblock {\em ACM Transactions on Design Automation of Electronic Systems},
  vol. 12, no. 4, pp. 42, 2007.

\bibitem{MBC:2008}
A.~Mishchenko, R.~Brayton, and S.~Chatterjee,
\newblock ``Boolean factoring and decomposition of logic networks,''
\newblock in {\em Proceedings of the 2008 IEEE/ACM International Conference on
  Computer-Aided Design}. IEEE Press, 2008, pp. 38--44.

\bibitem{WD:2010}
R.~Wille and R.~Drechsler,
\newblock {\em Towards a Design Flow for Reversible Logic.},
\newblock Springer, 2010.

\bibitem{SSP:2013}
A.~Shafaei, M.~Saeedi, and M.~Pedram,
\newblock ``Reversible logic synthesis of $k$-input, $m$-output lookup
  tables,''
\newblock in {\em In: Proceedings of the Conference on Design, Automation and
  Test in Europe (DATE'13)}, 2013, pp. 1235--1240.

\bibitem{LJ:2014}
C.-C. Lin and N.~K. Jha,
\newblock ``{RMDDS: Reed-Muller decision diagram synthesis of reversible logic
  circuits},''
\newblock {\em ACM Journal on Emerging Technologies in Computing Systems}, vol.
  10, no. 2, pp. 14, 2014.

\bibitem{MS:2011}
D.~Maslov and M.~Saeedi,
\newblock ``{Reversible circuit optimization via leaving the Boolean domain},''
\newblock {\em IEEE Trans. on CAD of Integrated Circuits and Systems}, vol. 30,
  no. 6, pp. 806--816, 2011.

\bibitem{MDM:2005}
D.~Maslov, G.~W. Dueck, and D.~M. Miller,
\newblock ``Toffoli network synthesis with templates,''
\newblock {\em IEEE Trans. on CAD of Integrated Circuits and Systems}, vol. 24,
  no. 6, pp. 807--817, 2005.

\bibitem{GM:2012}
O.~Golubitsky and D.~Maslov,
\newblock ``{A study of optimal 4-bit reversible Toffoli circuits and their
  synthesis},''
\newblock {\em IEEE Trans. Computers}, vol. 61, no. 9, pp. 1341--1353, 2012.

\bibitem{SM:2013}
M.~Saeedi and I.~L. Markov,
\newblock ``Synthesis and optimization of reversible circuits - a survey,''
\newblock {\em ACM Comput. Surv.}, vol. 45, no. 2, pp. 21, 2013.

\bibitem{NC:2000}
M.~A. Nielsen and I.~L. Chuang,
\newblock {\em Quantum Computation and Quantum Information},
\newblock Cambridge University Press, Cambridge, UK, 2000.

\bibitem{Bennett:89}
C.~H. Bennett,
\newblock ``Time/space trade-offs for reversible computation,''
\newblock {\em SIAM Journal on Computing}, vol. 18, pp. 766--776, 1989.

\bibitem{ALS+:2007}
A.~V. Aho, M.~S. Lam, R.~Sethi, and J.~D. Ullman,
\newblock {\em Compilers: Principles, Techniques, and Tools},
\newblock Addison Wesley, 2007.

\bibitem{Chan:2013}
S.~M. Chan,
\newblock {\em Pebble games and complexity},
\newblock Ph.D. thesis, Electrical Engineering and Computer Science, UC
  Berkeley, 2013,
\newblock Tech report: EECS-2013-145.

\bibitem{BTV:2001}
H.~Buhrman, J.~Tromp, and P.~M.~B Vit\'{a}nyi,
\newblock ``Time and space bounds for reversible simulation,''
\newblock in {\em Proc. ICALP 2001}, 2001, pp. 1017--1027.

\bibitem{LMT:2000}
K.~J. Lange, P.~McKenzie, and A.~Tapp,
\newblock ``Reversible space equals deterministic space,''
\newblock {\em J. Comput. Syst. Sci.}, vol. 60, no. 2, pp. 354--367, 2000.

\bibitem{LS:90}
R.~Y. Levine and A.~T. Sherman,
\newblock ``{A note on Bennett's space-time tradeoff for reversible
  computation},''
\newblock {\em SIAM J. Comput.}, vol. 19, no. 4, pp. 673--677, 1990.

\bibitem{FA:2001}
Frank.~M. P. and M.~J. Ammer,
\newblock ``{Separations of reversible and irreversible space-time complexity
  classes},''
\newblock Manuscript available from
  \texttt{http://www.eng.fsu.edu/~mpf/revsep.pdf}.

\bibitem{Knill:95}
E.~Knill,
\newblock ``{An analysis of Bennett's pebble game},''
\newblock arXiv.org preprint quant-ph/9508218.

\bibitem{Muchnick:97}
S.~S. Muchnick,
\newblock {\em Compiler design and implementation},
\newblock Morgan Kaufmann, 1997.

\bibitem{SGC:2012}
D.~Syme, A.~Granicz, and A.~Cisternino,
\newblock {\em {Expert F$\#$ 3.0}},
\newblock Apress Publishing, 2012.

\bibitem{DKR+06}
Th.~G Draper, S.~A Kutin, E.~M Rains, and K.~M Svore,
\newblock ``A logarithmic-depth quantum carry-lookahead adder,''
\newblock {\em Quantum Information \& Computation}, vol. 6, no. 4, pp.
  351--369, 2006.

\bibitem{CDKM:2004}
S.~A. Cuccaro, Th.~G Draper, S.~A. Kutin, and D.~P. Moulton,
\newblock ``A new quantum ripple-carry addition circuit,''
\newblock arXiv preprint quant-ph/0410184, 2004.

\bibitem{BCD+:96}
D.~Beckman, A.~N. Chari, S.~Devabhaktuni, and J.~Preskill,
\newblock ``Efficient networks for quantum factoring,''
\newblock {\em Phys. Rev. A}, vol. 54, pp. 1034--1063, 1996.

\bibitem{blif}
``{Berkeley Logic Interchange Format (BLIF)},''
\newblock Available at
  \texttt{https://www.ece.cmu.edu/~ee760/760docs/blif.pdf}, 1992.

\bibitem{benches}
K.~Minkovich,
\newblock ``{BLIF benchmark suite},''
\newblock Available at \texttt{http://cadlab.cs.ucla.edu/~kirill/}.

\bibitem{MCNC91}
S.~Yang,
\newblock ``{Logic synthesis and optimization benchmarks user guide},''
\newblock Available at
  \texttt{http://ddd.fit.cvut.cz/prj/Benchmarks/LGSynth91.pdf}.

\bibitem{LGSynth91}
J.~Mohnke,
\newblock ``A signature-based approach to formal logic verification,''
\newblock Available at
  \texttt{http://sundoc.bibliothek.uni-halle.de/diss-online/99/99H050/t8.pdf}.

\bibitem{GJ:79}
M.~R. Garey and D.~S. Johnson,
\newblock {\em Computers and Intractability: A Guide to the Theory of
  NP-Completeness},
\newblock W.H. Freeman, 1979.

\bibitem{GHH+:96}
O.~Goldschmidt, D.~S. Hochbaum, C.~A.~J. Hurkens, and G.~Yu,
\newblock ``Approximation algorithms for the $k$-clique covering problem,''
\newblock {\em SIAM J. Discrete Math}, vol. 9, no. 3, pp. 492--509, 1996.

\bibitem{KPF:2006}
L.~A.~B. Kowada, R.~Portugal, and C.~M.~H. {de Figueiredo},
\newblock ``Reversible karatsuba's algorithm,''
\newblock {\em Journal of Universal Computer Science}, vol. 12, no. 5, pp.
  499--511, 2006.

\bibitem{SS:2015}
K.~Shepley and R.~Steinwandt,
\newblock ``{Quantum circuits for $GF(2^n)$-multiplication with subquadratic
  gate count},''
\newblock {\em Quantum Information Processing}, vol. 14, no. 7, pp. 2373--2386,
  2015.

\bibitem{HP:90}
J.~L. Hennessey and D.~A. Patterson,
\newblock {\em Computer architecture: a quantitative approach},
\newblock Morgan Kauffmann, 1990.

\bibitem{SHA}
``Federal information processing standards publication 180-2,'' 2002,
\newblock See also the Wikipedia entry http://en.wikipedia.org/wiki/SHA-2.

\bibitem{MOV:97}
A.~J. Menezes, P.~C. van Oorschot, and S.~A. Vanstone,
\newblock {\em Handbook of Applied Cryptography},
\newblock CRC Press, 1997,
\newblock See also the Wikipedia entry http://en.wikipedia.org/wiki/MD5.

\end{thebibliography}

\appendices

\newpage
\onecolumn{
\section{Implementation of depth optimized adders in~\REVS}\label{app:adders}

In the following we give the details of a depth-optimized adder that is obtained from a standard classical construction which is then subsequently mapped to a reversible circuit using the compilation strategies available in \REVS.
A so-called carry-select adder \cite{HP:90} implements the addition of two $n$-bit integers in depth $O(\sqrt{n})$. The basic idea is to decompose the $n$ bits into $n/k$ blocks of size $k$ each, then to perform an addition for each block separately with two adders, one for each possible value of the incoming carry.
This leads to a doubling of the hardware cost plus the cost for multiplexers to select the correct sequence of adders for the given inputs, however, it also leads to a decrease in circuit depth as both branches can be synthesized for fixed value of the incoming carry and can be executed in parallel.
By choosing the block size to be $k=O(\sqrt{n})$ it can be shown that $O(\sqrt{n})$ depth can be achieved using a circuit size that still scales linear with $n$. A basic F$\#$ implementation of a carry select adder is given below.

\begin{lstlisting}[language=FSharp]
let carrySelectAdder n =
    let adderSize = int (sqrt (float n))
    let imSpacing = 2*(adderSize+1)
    <@@
        let carryRipple (a:bool array) (b:bool array) (carry:bool) =
            let result =  Array.zeroCreate (adderSize + 1)
            result.[0] <- a.[0] <> b.[0]
            let mutable carry = a.[0] && b.[0]
            result.[1] <- a.[1] <> b.[1] <> carry
            for i in 2 .. adderSize - 1 do
                carry <- (a.[i-1] && (carry <> b.[i-1])) <> (carry && b.[i-1])
                result.[i]  <-  a.[i] <> b.[i] <> carry
            result
        let a = Array.zeroCreate (adderSize * adderSize)
        let b = Array.zeroCreate (adderSize * adderSize)
        let a0b0 = carryRipple a.[0..adderSize-1] b.[0..adderSize-1] false
        let a' = a.[adderSize .. adderSize*adderSize-1]
        let b' = b.[adderSize .. adderSize*adderSize-1]
        let mutable intermediateResults = Array.zeroCreate (0)
        for i in 0 .. adderSize - 2 do
            intermediateResults <- Array.append intermediateResults
			(carryRipple a'.[i*adderSize .. i*adderSize + adderSize-1]
						  b'.[i*adderSize .. i*adderSize + adderSize-1] false)
            intermediateResults <- Array.append intermediateResults
			(carryRipple a'.[i*adderSize .. i*adderSize + adderSize-1]
						  b'.[i*adderSize .. i*adderSize + adderSize-1] true )
        let mutable result =  Array.zeroCreate (0)
        result <- Array.append result (a0b0.[0..adderSize-1])
        let mutable carry = a0b0.[adderSize]
        for i in 0 .. adderSize - 2 do
            let sum =
                if carry then
                    carry <- intermediateResults.[i*imSpacing+adderSize]
                    intermediateResults.[i*imSpacing..i*imSpacing+adderSize-1]
                else
                    carry <- intermediateResults.[i*imSpacing+2*adderSize+1]
                    intermediateResults.
                         [i*imSpacing+adderSize+1..i*imSpacing+2*adderSize]
            result <- Array.append result sum
        result
      @@>
\end{lstlisting}

\newpage
\section{Depth optimized adders in \REVS: resource estimates}\label{app:addersResources}

\begin{table*}[hbt]
\centering
\begin{tabular}{rrrrrrrrr}
\toprule
& \multicolumn{2}{c}{Hand optimized} & \multicolumn{3}{c} {Bennett cleanup}  &  \multicolumn{3}{c} {Eager cleanup} \\
\cmidrule(r){2-3} \cmidrule(r){4-6} \cmidrule(r){7-9}
$n$ & $\#$gates & $\#$qubits & $\#$gates & $\#$qubits & time &
$\#$gates & $\#$qubits & time \\
\midrule
10 & 29     & 34        & 54    & 68       & 0.0262 & 76    & 46       & 0.1132  \\
15 & 51     & 53        & 54    & 68       & 0.0232 & 76    & 46       & 0.0242  \\
20 & 73     & 72        & 118   & 123      & 0.0312 & 194   & 78       & 0.0332  \\
25 & 101    & 93        & 206   & 194      & 0.0402 & 344   & 118      & 0.0452  \\
30 & 120     & 111        & 206   & 194      & 0.0422 & 344   & 118      & 0.0452  \\
35 & 148     & 132        & 206   & 194      & 0.0402 & 344   & 118      & 0.0442  \\
40 & 167     & 150        & 318   & 281      & 0.0532 & 566   & 166      & 0.0542  \\
\bottomrule
\end{tabular}
\caption{\label{tab:adderDepth} Comparison of different compilation strategies for $n$-bit adders that are optimized for overall circuit depth. Shown are the results for a hand-optimized quantum carry lookahead adder and two adders that results from applying the~\REVS{} compiler to a classical depth optimized carry select adder with respect to a cleanup strategy corresponding to Bennett's method and a with respect to the eager cleanup strategy. Observe that the overall space requirement for the quantum circuits derived from the carry save arithmetic increases in a `plateau'-like fashion which is due to the usage of smaller size carry ripple adders that have a number of bits of size $O(\lceil \sqrt{n} \rceil)$.
Also observe that unlike Table \ref{tab:adderSize} here the number of gates differs between the three methods with the hand-optimized version being lowest, then Bennett's cleanup method, followed by the eager cleanup which has the highest gate counts throughout. However, the space requirements for the eager cleanup are better throughout than Bennett's method, and for some values of $n$ even better than the hand-optimized one, i.e., the eager cleanup strategy presents a possible space-time trade-off between circuit size and total number of qubits used. Like in case of the size optimized adders, the compilation times, measured in seconds, are comparable between the Bennett and eager cleanup strategies.}
\end{table*}

\onecolumn{
\section{Implementation of hash functions in \REVS}\label{app:hash}
\subsection{Implementation of SHA-2 in \REVS}

We already presented the core part of the SHA-2 hash function family in Section \ref{sec:data}.
Here we present an implementation of an entire algorithm for computing the entire round functions of the SHA-256 which is a member of the SHA-2 family that hashes a bitstring of arbitrary length to a bitstring of length $256$.
Our implementation actually only implements the round functions, which is the computationally and cryptographically most important part of the cipher, and not the message expansion step.
To describe the round functions, following \cite{SHA} it is convenient to introduce $8$ registers of $32$ bits each and to denote them by $A$, $B$, \ldots, $E$. Further, the following Boolean functions are introduced to describe the round functions:
\begin{align*}
{Ch}(E,F,G) &:= (E \wedge F) \oplus (\neg E \wedge G)\\
{Ma}(A,B,C) &:= (A \wedge B) \oplus (A \wedge C) \oplus (B \wedge C)\\
\Sigma_0(A) &:= (A\!\ggg\!2) \oplus (A\!\ggg\!13) \oplus (A\!\ggg\!22)\\
\Sigma_1(E) &:= (E\!\ggg\!6) \oplus (E\!\ggg\!11) \oplus (E\!\ggg\!25)
\end{align*}

For a given round, the values of all these functions is computed and considered to be $32$ bit integers.
Further, a constant $32$ integer value $K_i$ is obtained from a table lookup which depends on the number $i$ of the given round, where $i \in \{0,\ldots,63\}$ and finally the next chunk of the message $W_i$ is obtained from the message after performing a suitable message expansion is performed as specified in the standard.
Finally, $H$ is replaced according to
\[
H \leftarrow H +
{Ch}(E,F,G) +
{Ma}(A,B,C) +
\Sigma_0(A) +
\Sigma_1(E) + K_i + W_i
\]
and then the cyclic permutation $A \leftarrow H, B \leftarrow A, \ldots, H \leftarrow G$ is performed. The following F$\#$ program performs the computation of the entire round function for a given number of rounds $n$.
\begin{lstlisting}[language=FSharp]
let exprSha n =
    <@
    let mutable k = Array.zeroCreate 32
    let mutable w = Array.zeroCreate 32
    let mutable a = Array.zeroCreate 32
    let mutable b = Array.zeroCreate 32
    let mutable c = Array.zeroCreate 32
    let mutable d = Array.zeroCreate 32
    let mutable e = Array.zeroCreate 32
    let mutable f = Array.zeroCreate 32
    let mutable g = Array.zeroCreate 32
    let mutable h = Array.zeroCreate 32

    let addMod2_32 (a :bool array) =
        let b = Array.zeroCreate 32
        let mutable c = false
        b.[0] <- b.[0] <> a.[0]
        c <- c <> a.[0]
        a.[0] <- a.[0] <> (c && b.[0])
        for i in 1 .. 30 do
            b.[i] <- b.[i] <> a.[i]
            a.[i-1] <- a.[i-1] <> a.[i]
            a.[i] <- a.[i] <> (a.[i-1] && b.[i])
        b.[31] <- b.[31] <> a.[31]
        b.[31] <- b.[31] <> a.[30]
        for i in 2 .. n - 1 do
            a.[32-i] <- a.[32-i] <> (a.[31-i] && b.[32-i])
            a.[31-i] <- a.[31-i] <> a.[32-i]
            b.[32-i] <- b.[32-i] <> a.[31-i]
        a.[0] <- a.[0] <> (c && b.[0])
        c <- c <> a.[0]
        b.[0] <- b.[0] <> c
        clean c
        b


\end{lstlisting}

\begin{lstlisting}[language=FSharp]
    let ch (e :bool array)  (f :bool array)  (g :bool array) =
        let t =  Array.zeroCreate 32
        for i in 0 .. 31 do
            t.[i] <- (e.[i] && f.[i]) <> ((e.[i] <> false) && g.[i])
        t

    let ma (a :bool array)  (b :bool array)  (c :bool array) =
        let t =  Array.zeroCreate 32
        for i in 0 .. 31 do
            t.[i] <- (a.[i] && (b.[i] <> c.[i])) <> (b.[i] && c.[i])
        t

    let s0 a =
        let a2 = rot 2 a
        let a13 = rot 13 a
        let a22 = rot 22 a
        let t =  Array.zeroCreate 32
        for i in 0 .. 31 do
            t.[i] <- a2.[i] <> a13.[i] <> a22.[i]
        t

    let s1 a =
        let a6 = rot 6 a
        let a11 = rot 11 a
        let a25 = rot 25 a
        let mutable t =  Array.zeroCreate 32
        for i in 0 .. 31 do
            t.[i] <- a6.[i] <> a11.[i] <> a25.[i]
        t

    for i in 0 .. n - 1 do
        // Inplace add functions which
        // Take an input and add onto
        // the output in the assignment
        h  <- addMod2_32 (ch e f g)
        h  <- addMod2_32 (s0 a)
        h  <- addMod2_32 w
        h  <- addMod2_32 k
        d  <- addMod2_32 h
        h  <- addMod2_32 (ma a b c)
        h  <- addMod2_32 (s1 e)
        let t = h
        //Reassignment for next loop iteration
        h <- g; g <- f; f <- e; e <- d
        d <- c; c <- b; b <- a; a <- t
    Array.concat [a;b;c;d;e;f;g;h]
    @>
\end{lstlisting}
}

\onecolumn{
\subsection{Implementation of MD5 in \REVS}

Another hash function that we implemented in~\REVS{} is the so-called MD5 hash function.
From the cryptographic standpoint, MD5 is no longer of interest as it is considered to be broken \cite{MOV:97}, however, we can still use it as an example to exercise the compiler as the building blocks used in the cipher are well-suited to demonstrate the ease with which a classical function can be turned into a reversible circuit using \REVS.
MD5 hashes a bitstring of arbitrary length to a bitstring of length 128 and, like SHA-256 in the previous section, the cipher consists of a simple round function that gets applied many times to the current internal state and the next bits from the input and a message expansion function that takes the incoming bitstream and partitions it into suitable chunks.
As in case of SHA-256, we focus on the round function and show how it can be implemented by means of a reversible circuit.
The 128 bit state of MD5 can be conveniently expressed using $4$ registers of $32$ bits each, denoted by $A$, $B$, $C$, and $D$. Furthermore, the following Boolean functions are introduced:
\begin{align*}
F(B,C,D) &:= (B \wedge C) \vee (\neg B \wedge D) \\
G(B,C,D) &:= (B \wedge D) \vee (C \wedge \neg D) \\
H(B,C,D) &:=  B \oplus C \oplus D \\
I(B,C,D) &:=  C \oplus (B \vee \neg D).
\end{align*}

For a given round of index $i$ precisely one of the functions $f_i \{ F, \ldots, I\}$ is chosen according to a fixed schedule, then the value $f(B,C,D)$ is computed and then $A$ is updated as
$ A \rightarrow A \oplus f(B,C,D) \oplus M_i \oplus K_i$ is computed, where $K_i$ are precomputed constants, and $M_i$ are the bits of the message after message expansion has been performed. Subsequently, a bit rotation to the left by $s_i$ positions, where $s_i$ again are precomputed constants, and a further xor sum with the $B$ register is performed and the overall result is stored in the $A$ register. Finally, a cyclic rotation $A \rightarrow D$, $B \rightarrow A$, $C \rightarrow B$, $D \rightarrow C$ is performed which is the result of the $i$th round. The following F$\#$ program performs the computation of the entire round function for a given number of rounds $n$.

\begin{lstlisting}[language=FSharp]
let exprMD5 n =
    let s = [| 7; 12; 17; 22;  7; 12; 17; 22;  7; 12; 17; 22;  7; 12; 17; 22;
               5;  9; 14; 20;  5;  9; 14; 20;  5;  9; 14; 20;  5;  9; 14; 20;
               4; 11; 16; 23;  4; 11; 16; 23;  4; 11; 16; 23;  4; 11; 16; 23;
               6; 10; 15; 21;  6; 10; 15; 21;  6; 10; 15; 21;  6; 10; 15; 21 |]
    <@
    let mutable M = Array.zeroCreate 512
    let K = Array.zeroCreate 32
    let mutable A = Array.zeroCreate 32
    let mutable B = Array.zeroCreate 32
    let mutable C = Array.zeroCreate 32
    let mutable D = Array.zeroCreate 32

    let addMod2_32 (a :bool array) =
        let b = Array.zeroCreate 32
        let mutable c = false
        b.[0] <- b.[0] <> a.[0]
        c <- c <> a.[0]
        a.[0] <- a.[0] <> (c && b.[0])
        for i in 1 .. 30 do
            b.[i] <- b.[i] <> a.[i]
            a.[i-1] <- a.[i-1] <> a.[i]
            a.[i] <- a.[i] <> (a.[i-1] && b.[i])
        b.[31] <- b.[31] <> a.[31]
        b.[31] <- b.[31] <> a.[30]
        for i in 2 .. n - 1 do
            a.[32-i] <- a.[32-i] <> (a.[31-i] && b.[32-i])
            a.[31-i] <- a.[31-i] <> a.[32-i]
            b.[32-i] <- b.[32-i] <> a.[31-i]
        a.[0] <- a.[0] <> (c && b.[0])
        c <- c <> a.[0]
        b.[0] <- b.[0] <> c
        clean c
        b

    let F (B : bool array)  (C : bool array)  (D : bool array) =
        let t =  Array.zeroCreate 32
        for i in 0 .. 31 do
            t.[i] <- (B.[i] && C.[i]) || ((B.[i] <> false) && D.[i])
        t

    let G (B : bool array)  (C : bool array)  (D : bool array) =
        let t =  Array.zeroCreate 32
        for i in 0 .. 31 do
            t.[i] <- (D.[i] && B.[i]) || ((D.[i] <> false) && C.[i])
        t

    let H (B : bool array)  (C : bool array)  (D : bool array) =
        let t =  Array.zeroCreate 32
        for i in 0 .. 31 do
            t.[i] <- (B.[i] <> C.[i] <> D.[i])
        t

    let I (B : bool array)  (C : bool array)  (D : bool array) =
        let t =  Array.zeroCreate 32
        for i in 0 .. 31 do
            t.[i] <- (C.[i] <> (B.[i] || (D.[i] <> false)))
        t
\end{lstlisting}

\begin{lstlisting}[language=FSharp]
    for i in 0 .. 15 do
        let mutable t = Array.zeroCreate 32
        t <- addMod2_32 A
        t <- addMod2_32 (F B C D)
        t <- addMod2_32 K
        t <- addMod2_32 M.[32*i..32*i+31]
        t <- rot s.[i] t
        B  <- addMod2_32 t
        let temp = D
        D <- C; C <- B; A <- temp

 	for i in 16 .. 31 do
        let mutable t = Array.zeroCreate 32
        t <- addMod2_32 A
        t <- addMod2_32 (G B C D)
        t <- addMod2_32 K
        t <- addMod2_32 M.[32*((5*i+1)%16)..32*((5*i+1)%16)+31]
        t <- rot s.[i] t
        B  <- addMod2_32 t
        let temp = D
        D <- C; C <- B; A <- temp

 	for i in 32 .. 47 do
        let mutable t = Array.zeroCreate 32
        t <- addMod2_32 A
        t <- addMod2_32 (H B C D)
        t <- addMod2_32 K
        t <- addMod2_32 M.[32*((3*i+5)%16)..32*((3*i+5)%16)+31]
        t <- rot s.[i] t
        B  <- addMod2_32 t
        let temp = D
        D <- C; C <- B; A <- temp

	for i in 48 .. 63 do
        let mutable t = Array.zeroCreate 32
        t <- addMod2_32 A
        t <- addMod2_32 (I B C D)
        t <- addMod2_32 K
        t <- addMod2_32 M.[32*((7*i)%16)..32*((7*i)%16)+31]
        t <- rot s.[i] t
        B  <- addMod2_32 t
        let temp = D
        D <- C; C <- B; A <- temp
    Array.concat [A; B; C; D]
    @>
\end{lstlisting}

\end{document}